\documentclass[dvipsnames]{article}
\usepackage{xcolor}
\usepackage{adjustbox}
\usepackage[ruled]{algorithm2e}
\usepackage{amsfonts}
\usepackage{amsmath}
\usepackage{amsthm}
\usepackage{authblk}
\usepackage[backend=biber,citestyle=authoryear-comp, sorting = nyt, minnames=1, maxnames = 10]{biblatex}
\usepackage{bigints}
\usepackage{bm}
\usepackage{booktabs}
\usepackage{ctable}
\usepackage{float}
\usepackage{graphicx}
\usepackage{geometry}
\usepackage{hyperref}

\usepackage{mathtools}
\usepackage{setspace}
\usepackage{siunitx}
\usepackage{subcaption}

\usepackage{hyperref} 

\hypersetup{
    colorlinks=true,
    citecolor = blue,
    linkcolor=blue,
    filecolor=magenta,      
    urlcolor=cyan,
    pdfpagemode=FullScreen,
    }
\makeatletter
\renewcommand\AB@affilsepx{ \quad } 
\makeatother
\newcommand\deq{\stackrel{d}{=}}
\newcommand{\allmu}{\Tilde{\bmu}}
\newcommand{\allsig}{\Tilde{\bSig}}
\newcommand{\alllam}{\Tilde{\blam}}
\newcommand{\alltheta}{\Tilde{\btheta}}
\newcommand{\allpi}{\Tilde{\bm{\pi}}}
\newcommand{\bA}{\bm{A}}
\newcommand{\half}{\frac{1}{2}}
\newcommand{\bDeltag}{\bm{\Delta}_g }

\newcommand{\blamg}{\bm{\lambda}_g }
\newcommand{\ku}{ \kappa }

\newcommand{\bmug}{\bm{\mu}_g }

\newcommand{\bSigg}{\bm{\Sigma}_g}

\newcommand{\bX}{\bm{X}}
\newcommand{\bXo}{\bm{X}^{(o)}}
\newcommand{\bXm}{\bm{X}^{(m)}}
\newcommand{\bx}{\bm{x}}

\newcommand{\bmu}{ \bm{\mu}}
\newcommand{\bSig}{ \bm{\Sigma}}
\newcommand{\blam}{ \bm{\lambda}}
\newcommand{\luni}{\lambda_0}
\newcommand{\bDelta}{ \bm{\Delta}}
\newcommand{\btheta}{ \bm{\theta}}
\newcommand{\bOmega}{ \bm{\Omega}}
\newcommand{\bOmegag}{ \bm{\Omega}_g}

\newcommand{\ev}{ \mathbb{E}}
\newcommand{\by}{\bm{y}}
\newtheorem{definition}{Definition}[section]
\newtheorem{theorem}{Theorem}[section]
\newtheorem{lemma}{Lemma}[section]

\title{Clustering data with values missing at random using scale mixtures of multivariate skew-normal distributions}
\date{}
\author[1]{Pillay J. }
\author[2]{Tortora C.}
\author[3]{Punzo A.  }
\author[1]{Bekker A. }
\affil[1]{Department of Statistics, University of Pretoria, Pretoria, South Africa}
\affil[2]{San José State University, Department of Mathematics and Statistics, San Jose, USA}
\affil[3]{University of Catania, Department of Economics and Business, Catania,  Italy}

\begin{document}

\maketitle

\begin{abstract}
   Handling missing data is a major challenge in model-based clustering, especially when the data exhibit skewness and heavy tails. We address this by extending the finite mixture of scale mixtures of multivariate  skew-normal (FMSMSN) family to accommodate incomplete data under a missing at random (MAR) mechanism. Unlike previous work that is limited to one of the special cases of the FMSMSN family, our method offers a cluster analysis methodology for the entire family that accounts for skewness and excess kurtosis amidst data with missing values. The multivariate skew-normal distribution, as parameterised by \cite{azzalini1996} and \cite{arnoldbeaver} includes the normal distribution as a special case, which ensures that our method is flexible toward existing symmetric model-based clustering techniques under a normality assumption. We derive the distributional properties of the missing components of the data and propose an augmented EM-type algorithm tailored for incomplete observations. The modified E-step yields closed-form expressions for the conditional expectations of the missing values. The simulation experiments showcase the flexibility of the FMSMSN family in both clustering performance and parameter recovery for varying percentages of missing values, while incorporating the effects of sample size and cluster proximity. Finally, we illustrate the practical utility of the proposed method by applying special cases of the FMSMSN family to global CO2 emissions data. 
\end{abstract}

\section{Introduction}
    Analysing high-dimensional data is now standard practice in industry and academia due to constant improvements in computational power and the increased availability of high-dimensional data. Model-based clustering is one of the more popular analysis tools because the results from the method are typically easy to interpret, and with the available computational power, the method is efficient in fitting the data. Finite mixture modelling is appropriate for describing heterogeneity in a population by assuming the said population is a mixture of sub populations, each of which can be modelled by a probability distribution. In theory, the finite mixture model’s design allows it to be fitted with any probability distribution needed. Historically mixtures of multivariate normal distributions have dominated the finite mixture modelling literature. These mixtures are often referred to as a Gaussian Mixture Models (GMMs). Their popularity is due to the algebraic and computational simplifications inherent to the multivariate normal distribution, which make it easier to fit data (recent examples of GMMs in applied fields include \cite{bayesianGMM, featuresSelectionGMM, chemistryGMM}). This, combined until recently with a small literary pool providing theory on finite mixture models with other distributions, left practitioners assuming normality out of need rather than aptness or belief. 
    
    There are datasets in many fields whose geometry is not sufficiently described by the elliptical symmetric normal distribution because they exhibit significant skewness. Noticeably, when cluster distributions are asymmetric and have heavier tails, GMMs often compensate by overestimating the number of clusters. The 'extra' clusters are not helpful in practice and weaken the practitioner's interpretations and analyses. The demand for a more flexible distribution that could handle asymmetry led to the development of the multivariate skew-normal distribution, first introduced by \cite{azzalini1996}. As the name suggests, the distribution's construction encompasses normality as a special case but is also embedded with an explicit skewness parameter that can describe asymmetry. In recent times the multivariate skew-normal distribution has proven useful in chemistry (\cite{chemparticle, chempolycrystal, quantumcrystals}), environmental, biomedical (\cite{lachoserrors, lachoslikelihood,skewnormalcentred}), and economic (\cite{insurance}) fields including  modelling toolkits such as linear regression, time series, and graphic models, (\cite{linreg}, \cite{covid19}). The multivariate skew-normal distribution still does not address the heavy tails some empirical distributions exhibit. This motivated the construction of a family of new probability distributions borne of the multivariate skew-normal to handle excess kurtosis in data, known as Scale Mixture of Skew-Normal (SMSN) distributions. When extended to a clustering context, a finite mixture of SMSN distributions thus is a more flexible model that can better describe asymmetrical clusters in data. Several analyses have benefited from the SMSN family, including applications of the multivariate skew-hyperbolic (\cite{skew-normalgeneralisedhyperbolic}), skew-t (\cite{windspeed, skewtschizo, skewtlnk}), and skew-power-exponential (\cite{skewpower}) distributions, among many other multivariate members (\cite{skewnormalcensor,lachos2011}).
    
    Model-based clustering using the FMSMSN family displays a strong classification performance on complete datasets but cannot be applied as is to datasets with missing values (motivating examples can be found in \cite{lachos2011, lachos2011,mcnicholas2013}). Incomplete datasets are a common challenge faced in numerous fields, including economic, environmental (\cite{missingspatial}), and, more commonly, public health and medical fields (\cite{missingmedical}). When data have missing values, practitioners usually either subset the complete data to fit a clustering method or use an array of imputation techniques to fill in the missing values and fit the same model on each imputation. The former discards valuable information, which could lead to a biased model fit and harms/weakens interpretations offered by the clustering results. The latter can be computationally expensive. It is thus worth pursuing imputation approaches to circumvent possible bias and variance suppression (\cite{missingbias, missingimpact, missingbayes}). From what was surveyed, the recent developments on accommodating missingness show a need and a varied use of special members of the SMSN family. Distributions such as the multivariate skew-normal and multivariate skew-t are among the popular members explored (\cite{skewnormalcensor, skewnormalEMcensor, skewtlnk, skewtcensor, skewtffmreg}). Many of these techniques assume a special property of missingness in data: the probability of a missing entry in an observed vector is not dependent on the value of said entry but potentially dependent on the value of the observed entries, otherwise known as Missing At Random (MAR) - see \cite{definemissing} for a discussion on apt definitions and properties of the types of missing patterns found in data. It is thus clear that the literature on extending model-based clustering for the FMSMSN family to fit to incomplete data is ongoing, but limited.


     This paper makes two primary contributions: one theoretical and one methodological. First, it extends the distribution theory of the multivariate skew-normal distribution explicitly in terms of a random vector’s observed and missing components. In this context, it allows us to derive the full probability distribution governing the Missing At Random (MAR) mechanism for all members of the FMSMSN family, a novel development not currently available in the surveyed literature. In doing so, it establishes crucial closed-form expressions for the expectation and covariance structures while also exploring how different members of the family relate to each other under the MAR assumption. The second major contribution is a generalised estimation approach for FMSMSN models with incomplete data. We develop a maximum likelihood estimation procedure based on the Expectation-Maximisation (EM) algorithm, specifically adapted to account for MAR data. Since fitting a finite mixture model is restructured as an incomplete data problem we enhance the estimation technique to simultaneously handle incomplete observations. The added benefit of the proposed algorithm is that a set of imputed values for the missing parts of the dataset is provided at the end of the algorithm with minimal additional computational cost. This dual functionality not only enhances parameter recovery but also provides a heuristic way to reconstruct incomplete datasets that preserves the underlying distributional structure.

     The rest of this paper is organised as follows: Section \ref{preliminaries} briefly introduces the multivariate skew-normal distribution and some of its relevant properties, followed by some special cases of the SMSN family members. The methodology and novelty of the paper is discussed in Section \ref{methodology}. Section \ref{sec:simulations} carries out an extensive simulation design to explore what effects the proportion of missing rows have on clustering performance and parameter recovery. Section \ref{application} applies the algorithm to global CO2 emissions data for the most recent year available from the EDGAR database for relevant insights and findings. The clustering results suggests a link between an improving Gross Domestic Product (GDP) and the increased carbon emissions as well as a distinguishing difference between the global south and north, suggesting different driving factors behind their respective carbon footprints. Supplementary results required for some of the derivations are given in the Appendix section.

    
\section{Preliminaries}
\label{preliminaries}
This section recalls the multivariate skew-normal distribution and its relevant properties to the topic. There are numerous parameterisations of the multivariate skew-normal distribution, with some of the popular parameterisations originating from \cite{azzalini1996, arnoldbeaver, gupta} and \cite{sahu}. This paper considers the parameterisation published by \cite{arnoldbeaver} given in definition \ref{sn def}:

\begin{definition}
\label{sn def}
 A random vector $\bX \in \mathbb{R}^p$ follows a  multivariate skew-normal distribution with a location vector $\bmu\in \mathbb{R}^p$, a positive definite scale matrix $\bSig$, a skewness vector $\blam \in \mathbb{R}^p $ and a threshold parameter $\luni$ if it has the following probability density function (pdf):
   \begin{align}
    f_{\text{SN}}(\bx; \bmu, \bSig,\blam,\luni) = \frac{1}{\Phi_1\left( \delta_{0}\right)}\bm{\phi}_p \left(\bx; \bmu, \bSig \right) 
    \Phi_1\left( \lambda_{0} + \blam^{\top} \bSig^{-1/2} (\bx - \bmu )\right),
    \label{sn}
   \end{align}
   where $\delta_{0} = \frac{\luni}{\sqrt{ 1 + \blam^{\top} \blam}}$, and $\bm{\phi}_p\left(\bx; \bmu,\bSig\right) = \frac{1}{ (2\pi)^{p/2} | \bSig|^{1/2} } e^{ -\half \left\{ \left( \bx - \bmu \right)^{\top}\bSig^{-1} \left( \bx - \bmu \right) \right\} }$ is the p-variate normal pdf, and $\Phi_1(\cdot)$ is the univariate standard normal cumulative distribution function (cdf).\newline
\end{definition}
When $\bX$ follows pdf \eqref{sn}, we write $\bX \sim SN\left(\bmu, \bSig, \blam,\luni \right)$. In \cite{arnoldbeaver} the derivation of the multivariate skew-normal distribution begins by truncating elements of a normally distributed random vector against a threshold $\lambda_0$. The threshold parameter $\lambda_0$ arises in the construction of the multivariate skew-normal distribution through a method that involves truncating a bivariate normal distribution—specifically, by retaining only those observations where one component exceeds a certain threshold. From a simulation perspective, $\lambda_0$ represents this cutoff, effectively filtering the data to induce asymmetry, which directly contributes to the skewness of the resulting distribution. However, because skewness can also be controlled by another parameter, namely $\blam$, $\lambda_0$ is typically set to zero to simplify the model and avoid complex interactions between the two parameters. In Sections \ref{sec:simulations} and \ref{application}, $\lambda_0$ is set to 0 to explore the behaviour and performance of $\blam$, and the notation  $\bX \sim SN\left(\bmu, \bSig, \blam \right)$ is used to denote that $\bX$ follows a multivariate skew-normal distribution with this special case.
Since this formulation of the multivariate skew-normal distribution is multivariate normal location-scale mixture distribution, there exists a stochastic representation: consider random variable $U$ with some distribution $H(u;\btheta)$ independent of $\bm{Z} \sim SN(\bm{0},\bm{I},\blam)$ and let $K$ be some link function of $U$ - that is, let $K = K(U)$. Define the random vector 
\begin{align}
    \label{sr scale mix}
    \bX = \bmu + \left(K\bSig\right)^{1/2}\bm{Z},
\end{align} 
Then $\bX|U=u \sim SN(\bmu,\kappa\bSig,\blam)$, where $\ku = K(u)$. The random variable $U$ is attached to the scale matrix $\bSig$ and is aptly referred to as the scaling variable. The pdf of a scale mixture distribution of $\bX$ is given as:

\begin{align}
    f_{\text{SMSN}}(\bx; \bmu, \bSig,\blam,\btheta) =  2\int_{u \in \mathcal{S} } \bm{\phi}_p \left(\bx; \bmu, \ku \bSig \right) 
    \Phi_{1}\left( \blam^{\top} \left( \ku \bSig\right )^{-1/2} (\bx - \bmu )\right)  dH(u;\bm{\theta}),
    \label{sm dist}
   \end{align}
   where $\mathcal{S}$ is the support of the distribution of $U$.
The first two moments of the scale mixture distribution are derived using the stochastic representation (\ref{sr scale mix}). 
\begin{theorem}
   Suppose the distribution of the random vector $\bX$ has a pdf given by (\ref{sm dist}). If $ \ev[K^{\frac{1}{2}}],\ev[K]< \infty$ then:
\begin{center}
\begin{tabular}{rcl}
    \label{sm moments}
    $\ev \left[\bX\right]$ &         = & $\bmu + \sqrt{\frac{2}{\pi}}\omega_{1}\bDelta, \text{ and } $\\
    $\mathrm{Cov}\left[\bX\right]$ & = & $\omega_{2} \bSig - \frac{2}{\pi} \omega_{1} \bDelta \bDelta^{\top} \omega_1$,
\end{tabular}
\end{center}
where $ \omega_{r} =  \ev[K^{\frac{r}{2}}], \bDelta =  \frac{\bSig^{1/2}\blam}{\sqrt{ 1 +\blam^{\top} \blam}} \iff \blam = \frac{\bSig^{-1/2} \bDelta }{ \sqrt{1 - \bDelta^{\top} \bSig^{-1} \bDelta} }.$
\end{theorem}
   The special cases of the pdf in (\ref{sm dist}) we consider are briefly discussed in the next subsection.
   
   \subsection{Special cases of the SMSN family}
   \label{special cases}
   Some examples of the distributions borne of the stochastic representation (\ref{sr scale mix}) and their properties are now presented 
   \subsubsection{Skew-t}
   Suppose $U$ has a gamma distribution defined in Appendix \ref{gamma dist def} with shape and rate parameters both equal to $\frac{\nu}{2}$.Here, $\nu$ is the degrees of freedom. That is, $U \sim Gam\left(\frac{\nu}{2}, \frac{\nu}{2} \right)$ and the link function $K = U^{-1}$. Then random vector $\bX \in \mathbb{R}^p$ has a multivariate skew-t distribution with parameters $\bmu, \bSig, \blam$, and $\nu$ and is denoted $\bX \sim ST(\bmu, \bSig, \blam, \nu )$ if it has the following pdf:
   \begin{align}
       f_{\text{ST}}(\bx; \bmu, \bSig, \blam, \nu) = 2t_p\left( \bx;\bmu, \bSig, \blam, \nu \right)T_1\left( \frac{ \sqrt{\nu + p} A }{\sqrt{d + \nu } } ; 0,1,\nu + p  \right),
   \end{align}
   where $t_p$ and $T_1$ respectively denote the pdf and cdf of a t-distribution, with $t_p$ defined in Appendix \ref{t dist def}, $A = \blam^{\top} \bSig^{-\half}(\bx - \bmu)$ and $d = (\bx - \bmu)^{\top}\bSig^{-1}(\bx - \bmu) $. As $\nu \rightarrow  \infty$, $\bX$ converges in distribution to a $ SN(\bmu, \bSig, \blam)$ distribution. The mean and covariance of $\bX$ are (\cite{skewtcensor}):
   \begin{center}
\begin{tabular}{rcl}
    \label{st moments}
    $\ev \left[\bX\right]$ &         = & $\bmu + \sqrt{ \frac{\nu}{\pi}  } \frac{ \Gamma\left( \frac{\nu-1}{2} \right) }{\Gamma \left(\frac{\nu}{2} \right) } \bDelta, \text{ provided that } \nu >1.$\\
    $\mathrm{Cov}\left[\bX\right]$ & = & $ \frac{\nu}{\nu -2} \bSig - \frac{\nu}{\pi}\left( \frac{ \Gamma\left( \frac{\nu - 1}{2} \right) }{\Gamma \left(\frac{\nu}{2} \right) } \right)^2 \bDelta \bDelta^{\top}, \text{ provided that } \nu >2.$
\end{tabular}
\end{center}

\subsubsection{Skew-slash}
Considering that $U \sim Beta(a,1) $ has a beta distribution as defined in Appendix \ref{beta pdf} with link function $K = U^{-1}$, $\bX$ has a multivariate skew-slash distribution with the following pdf:
\begin{align}
    f_{\text{SS}}(\bx) = 2\int_{0}^1 \alpha u ^{ \alpha -1} \bm{\phi}_p \left(\bx; \bmu, u^{-1} \bSig \right) 
    \Phi_{1}\left( \blam^{\top} \left( u^{-1}\bSig\right )^{-1/2} (\bx - \bmu )\right) du
\end{align}
which does not have a closed form. Similar to the multivariate skew-t distribution, $\bX$ converges in distribution to a $ SN(\bmu, \bSig, \blam)$ distribution as $\alpha \rightarrow \infty$. The mean and covariance of $\bX$ are (\cite{betamoments}):
   \begin{center}
\begin{tabular}{rcl}
    $\ev \left[\bX\right]$ &         = & $\bmug + \sqrt{ \frac{2}{\pi}  } \frac{ 2\alpha }{ 2\alpha -1 } \bDelta, \text{ provided that } \alpha > \half.$\\
    $\mathrm{Cov}\left[\bX\right]$ & = & $ \frac{\alpha}{\alpha -1} \bSig - \frac{2}{\pi}\left( \frac{ 2\alpha }{ 2\alpha -1 }\right)^2 \bDelta \bDelta^{\top}, \text{ provided that } \alpha >1.$
\end{tabular}
\end{center}

\subsubsection{Skew-variance-gamma}
For $U \sim Gam(\eta,\frac{\gamma^2}{2}) $ with link function $K = U$, $\bX \in \mathbb{R}^p$ has a skew-variance-gamma distribution with the following pdf (\cite{skewvgl}):
\begin{align}
    f_{\text{SVG}}(\bx) =  2 \left( \frac{\gamma}{2\pi} \right)^{p/2}\frac{ \gamma ^{\eta} }{\Gamma(\eta) 2^{\eta -1}} |\bSig|^{-1/2} \sqrt{d}^{\left( \eta - \frac{p}{2} \right) }K_{\eta - \frac{p}{2}}\left(\sqrt{d\gamma} \right)GH \left(A; \eta-\frac{p}{2}, \sqrt{d}, \gamma \right),
\end{align}
where $GH\left(\cdot; \eta-\frac{p}{2}, \sqrt{d}, \gamma \right)$ is the cdf of a generalised hyperbolic distribution defined in Appendix \ref{gh pdf}. 
\begin{center}
\begin{tabular}{rcl}
    $\ev \left[\bX\right]$ &         = & $\bmu + \frac{2}{\sqrt{\pi} \gamma} \frac{ \Gamma\left( \frac{2\eta +1}{2} \right) }{ \Gamma\left(\eta \right) }\bDelta, \text{ and } $\\
    $\mathrm{Cov}\left[\bX\right]$ & = & $ \frac{2\eta}{\gamma^2} \bSig - \frac{4}{\pi \gamma^2}\left( \frac{ \Gamma\left( \frac{2\eta +1}{2} \right) }{ \Gamma\left(\eta \right) } \right)^2 \bDelta \bDelta^{\top} $,
\end{tabular}
\end{center}
A special case of the multivariate skew-variance-gamma distribution arises when $\eta = 1$, in which case the pdf simplifies to the following:
\begin{align}
    f_{\text{SLP}}(\bx) = 2 \left( \frac{\gamma^{p/2+1} }{(2\pi)^{p/2}  } \right) |\bSig|^{-1/2} \sqrt{d}^{\left( 1 - \frac{p}{2} \right) }K_{1- \frac{p}{2}}\left(\sqrt{d\gamma} \right)GH \left(A; 1-\frac{p}{2}, \sqrt{d}, \gamma \right),
\end{align}
which corresponds to the pdf of a multivariate skew-Laplace distribution. Note that estimating the multivariate skew-variance-gamma distribution's parameters lead to identifiability issues (\cite{vg}) hence an additional constraint is added to remedy this. Moving on in this paper, $\ev[U] = 1 \implies \eta =\frac{\gamma^2}{2}$.

 In summary, the four special cases considered in this paper are:
   \begin{table}[H]
   \centering
    \caption{Examples of distributions from the SMSN family from \eqref{sr scale mix}.}
   \begin{tabular}{lcll}
   \toprule
    Distribution & Denoted as & $\ku$ & Distribution of $U$\\
    \midrule
    Multivariate skew-normal         & $ SN(\bmu, \bSig, \blam)$                  & $ u $     & Degenerate at $u=1$ \\
    Multivariate skew-t              & $ ST(\bmu, \bSig, \blam, \nu) $            & $u^{-1}$  & $U \sim Gam(\frac{\nu}{2},\frac{\nu}{2} )$ \\
    Multivariate skew-slash          & $ SS(\bmu, \bSig, \blam, \alpha)$          & $u^{-1}$  & $U \sim Beta(\alpha, 1) $ \\
    Multivariate skew-variance-gamma & $ SVG(\bmu, \bSig, \blam, \eta, \gamma^2)$ & $ u$      & $U \sim Gam\left(\eta, \frac{\gamma^2}{2} \right)$\\
    \bottomrule
       \end{tabular}
       \label{4cases}
   \end{table}
   The pdf (\ref{sm dist}) is now used to formally define a finite mixture of $G$ scale mixtures of multivariate skew-normal distributions.
\begin{definition}
\label{fmsmsnpdf}
     A random vector $\bX \in \mathbb{R}^p$ belongs to a Finite Mixture of Scale Mixture of Skew-normal (FMSMSN) distributions if it has the following weighted sum of pdfs:
   \begin{align}
   \sum_{g=1}^G \pi_g f_{\text{SMSN}}(\bx; \bmug, \bSigg,\blamg,\btheta_g),
    \label{fmm}
   \end{align}
   where $\pi_g >0$ is the mixing probability for the $g^{th}$ distribution subject to the constraint $\sum_{g=1}^G \pi_g=1$.
\end{definition}

\section{FMSMSN with MAR observations}
\label{methodology}
This section contains the novelty of this paper, namely, the extension of finite mixtures of scale mixtures of the multivariate skew-normal distribution to handle MAR observations.
The section unfolds in three parts. First, it formalises the probabilistic structure of the missing components under the MAR assumption. The second introduces inference, explaining how a member of the FMSMSN family is fitted to data. It highlights that the model parameters cannot be estimated tractably through direct optimisation of the log-likelihood, as both latent variables and missing data must be handled simultaneously during estimation. The first two parts provide the foundation that enables the inference that is the third part: the development of a modified and updated EM algorithm that incorporates probabilistic rules discussed in the first two parts to optimise the observed log-likelihood. The result is a new algorithm capable of fitting FMSMSN models to MAR data, addressing the identified limitation in the current literature. 
 
 Suppose $\bX_i$ is the $i^{th}$ observation distributed by model (\ref{fmm}). It is algebraically tangible to express $\bX_i$ as the following stochastic representation for the $g^{th}$ distribution (\cite{arnoldbeaver}):
        \begin{align}
            \bX_i \deq  \bmug +T_i \bm{\bDeltag} +  \sqrt{K_i}\bSigg^{1/2}(\bm{I} -\bm{\delta}_g \bm{\delta}_g^{\top} )^{1/2} \bm{T}_{1,i},
            \label{sr}
        \end{align}
   where $T_i= \sqrt{K_i}T_{0,i}, ~\bDeltag = \frac{\bSigg^{1/2}\blamg}{ \sqrt{ 1 + \blamg^{\top}\blamg}  } $, $K_i = K(U_i)$, $\bm{\delta}_g = \frac{\blamg}{ \sqrt{ 1 + \blamg^{\top}\blamg}  } $, random vector $\bm{T}_{1,i} \sim SN(\bm{0}, \bm{I},\bm{0})$ independent of univariate random variable $T_{0,i}$ which follows a truncated normal distribution on the interval $(0,\infty)$. In general, a truncated normal distribution with support $(0,\infty)$ has a pdf defined as follows:
   \begin{definition}
   \label{trunc normal}
       A random variable $T$ follows a truncated normal distribution with support $(0, \infty)$ and parameters $\mu$ and  $\sigma^2$, denoted as $T \sim TN(\mu,\sigma^2)$, if it has the following pdf:
       \begin{align}
            f_{\text{TN}}(t; \mu,\sigma^2) = \begin{cases}
                                         \frac{\phi_1(t; \mu,\sigma^2) }{\Phi_1\left( \frac{\mu}{\sigma}\right) } & t >0 \\
                                         0                                                                        & \text{otherwise},
                                      \end{cases}
       \end{align}
\end{definition}
 Note that the first two moments of $T$ are:
    \begin{center}
    \begin{tabular}{rcl}
        \label{trunc norm moments}
        $\ev [T] $ & = & $\mu + \sigma W_{\phi}\left( \frac{\mu}{\sigma}\right) \text{ and } $\\
        $\mathrm{var}[T]$  & = & $\sigma^2 \left( 1 -  W_{\phi}^2\left( \frac{\mu}{\sigma}\right)\right)$,
    \end{tabular}
    \end{center}
 where $ W_{\phi}(\cdot) = \frac{ \phi_1(\cdot;0,1)}{\Phi_1(\cdot)}.$
Next, an allocation vector, $\bm{Z}$, is introduced. Define random the variable: 
   \begin{align}
   Z_g = \begin{cases}
    1 & \text{if $\bX$ belongs to $g^{th}$ distribution with probability $\pi_g$}\\
    0 & \text{otherwise}
   \end{cases}
   \end{align}
   such that $\bm{Z} = (Z_1, \dots, Z_G)$. Then $\bm{Z}$ follows a multinomial distribution with the collection of parameters $\allpi = (\pi_1, \dots, \pi_G)$. Consequently, conditional on $\bm{Z} = \bm{z}$, the random variable $\bX$ follows the distribution of the $g^{th}$ component of the mixture model.

\subsection{Distribution properties under an MAR mechanism}
Hereon, $\bX_i$ is decomposed into its missing and observed vectors $\bX_i=\begin{pmatrix} \bXm_i\\ \bXo_i \end{pmatrix}$. The superscripts $o$ and $m$ indicate the observed and missing parts of $\bX_i$, respectively. While the missingness pattern may vary across observations, for notational convenience we write $o$ and $m$ in place of $o_i$ and $m_i$. From the MAR mechanism, some distribution properties of $\bXm_i$ are developed below in Theorems \ref{snm} to \ref{norm cond dist}.

 \begin{theorem}
 \label{snm}
   Suppose $ \bX_i \in \mathbb{R}^p$ has pdf \eqref{fmm} and is described as in \eqref{sr scale mix}. It follows that $\bX_i|U_i =u_i, Z_{i,g}=1 \hspace{0.15cm}\sim SN(\bmug, \ku_i \bSigg, \blamg )$ and
   $\bX_i|T_i=t_i, U_i =u_i, Z_{i,g}=1 \hspace{0.15cm}\sim SN(\bmug + t_i\bDeltag, \ku_i \bOmegag, \bm{0} )$, where $\bOmegag = \bSigg - \bDeltag \bDeltag^{\top}$.
   \begin{proof}
       First notice from the stochastic representation in (\ref{sr scale mix}) that $\bX_i|U_i =u_i, Z_{i,g}=1 \hspace{0.15cm}\sim SN(\bmug , \ku_i \bSigg, \blamg )$. The stochastic representation (\ref{sr}) implies that $\bX_i|T_i=t_i, U_i =u_i, Z_{i,g}=1$ follows a normal distribution with the following parameters:
       \begin{align}
         \ev[\bX_i|T_i, U_i, Z_{i,g}] &= \mu_g + t_i\bDeltag &\nonumber\\
         \mathrm{Cov}(\bX_i|T_i, U_i, Z_{i,g}) &= \ku_i\bSigg^{1/2}(\bm{I} - \bm{\delta}_g \bm{\delta}_g^{\top} )\bSigg^{1/2} = \ku_i(\bSigg - \bDeltag \bDeltag^{\top}) = \ku_i\bOmegag, &
       \end{align}
       which completes the proof.
   \end{proof}
   \end{theorem}

   \begin{theorem} Consider $ \bX_i \in \mathbb{R}^p$ as described in \eqref{sr scale mix}. Partition the random vector $\bX_i$ and its distribution parameters in terms of the $\bX_i$'s observed and missing components:
        \begin{align}
        \label{partitions}
        \bX_i = \begin{pmatrix} \bXm_i \\ \bXo_i \end{pmatrix},~
        \bmug = \begin{pmatrix} \bmu_{o,g} \\ \bmu_{m,g} \end{pmatrix},~ \bSigg = \begin{pmatrix} \bSig_{mm,g} & \bSig_{mo,g}\\ \bSig_{om,g} & \bSig_{oo,g}\end{pmatrix},~
        \blam= \begin{pmatrix} \blam_{m,g} \\ \blam_{o,g} \end{pmatrix}, \text{ and } \bDeltag= \begin{pmatrix} \bDelta_{m,g} \\ \bDelta_{o,g} \end{pmatrix},
    \end{align}
    where $\bDeltag$ is defined as in Theorem \ref{sm moments}. Then it is the case that:
   \label{sn cond dist}
    \begin{align}\bXm_i | \bXo_i=\bx^{(o)}_i, U_i =u_i, Z_{i,g}=1 \hspace{0.15cm} \sim SN(\bmu_{c,g}, \ku_i\bSig_{c,g}, \blam_{c,g}, \ku_i^{-\half}\lambda_{0,c,g}), \nonumber    
    \end{align}
    where:
    \begin{align}
     &\bmu_{c,g} = \bmu_{m,g} + \bSig_{mo,g} \bSig_{oo,g}^{-1}(\bx_{i}^{(o)} - \bmu_{o,g}),\hspace*{0.35cm}
     \bSig_{c,g} = \bSig_{mm} - \bSig_{mo,g} \bSig_{oo,g}^{-1}\bSig_{om,g},\hspace*{0.35cm} 
    \lambda_{0,c,g} = \frac{ \bDelta_{o,g}^{\top}\bSig_{oo,g}^{-1}(\bx_{i}^{(o)} - \bmu_{o,g}) }{\sqrt{1 - \bDelta_{g}^{\top}\bSigg^{-1}\bDelta_{g}} }\text{, and }&\nonumber\\
   & \blam_{c,g} = \frac{  \bSig_{c,g}^{-1/2} \left[\bDelta_{m,g} - \bSig_{mo,g} \bSig_{oo,g}^{-1} \bDelta_{o,g} \right] }{\sqrt{ 1 - \bDelta_{g}^{\top} \bSigg^{-1} \bDelta_{g}  } }.&
     \end{align}
   \end{theorem}
   \begin{proof}
        A detailed proof can be found in \cite{MarginalConditionalsn}.
   \end{proof}

   \begin{theorem}
   \label{norm cond dist}
    Consider $ \bX_i \in \mathbb{R}^p$ as described in \eqref{sr scale mix}. Partition the random vector $\bX_i$ and its distribution parameters in terms of the $\bX_i$'s observed and missing components:
        \begin{align}
        \bX_i = \begin{pmatrix} \bXm_i \\ \bXo_i \end{pmatrix}, ~
        \bmug = \begin{pmatrix} \bmu_{o,g} \\ \bmu_{m,g} \end{pmatrix},~ \bOmegag = \begin{pmatrix} \bOmega_{mm,g} & \bOmega_{mo,g}\\ \bOmega_{om,g} & \bOmega_{oo,g}\end{pmatrix}, \text{ and } \bDeltag= \begin{pmatrix} \bDelta_{m,g} \\ \bDelta_{o,g} \end{pmatrix},
    \end{align} Then it is the case that:
        \begin{align}
        \bXm_i | \bXo_i=\bx^{(o)}_i, T_i =t_i, U_i =u_i, Z_{i,g}=1 \hspace{0.15cm} \sim SN(\bm{m}_{c,g} + t_i \bm{\psi}_{c,g}, \ku_i\bOmega_{c,g}, \bm{0}),\nonumber
        \end{align}
         where:
        \begin{align}
        &\bm{m}_{c,g} = \bmu_{m,g} + \bOmega_{mo,g}\bOmega_{oo,g}^{-1}(\bx_{i}^{(o)} - \bmu_{o,g}), \hspace{0.25cm}   \bm{\psi}_{c,g}= \bDelta_{m,g} - \bOmega_{mo,g}\bOmega_{oo,g}^{-1} \bDelta_{o,g}, \hspace{0.25cm} \text{and } \hspace{0.25cm} \bOmega_{c,g} = \bOmega_{mm,g} - \bOmega_{mo,g}\bOmega_{oo,g}^{-1}\bOmega_{om,g}.&\nonumber
        \end{align}
   \end{theorem}
   \begin{proof}
       All that is needed is to notice from Theorem \ref{snm} that  $\bX_i|T_i=t_i, U_i =u_i, Z_{i,g}=1 \hspace{0.15cm}\sim SN(\bmug + t_i\bDeltag, \ku_i \bOmegag, \bm{0} )$. 
       It thus follows that $\bXm_i | \bXo_i=\bx^{(o)}_i, T_i =t_i, U_i =u_i, Z_{i,g}=1 \sim SN(\bmu_g^*, \bOmegag^*, \bm{0} )$ where,
         \begin{align}
          \bmug^*    &= \bmu_{m,g} + t_i\bDelta_{m,g} + (\ku_i\bOmega_{mo,g} ) (\ku_i \bOmega_{oo,g})^{-1}(\bx_{i}^{(o)} - \bmu_{o,g} - t_i\bDelta_{o,g})&  \nonumber \\
                     &= \bmu_{m,g} + \bOmega_{mo,g}\bOmega_{oo,g}^{-1}(\bx_{i}^{(o)} - \bmu_{o,g}) + t_i(\bDelta_{m,g} - \bOmega_{mo,g}\bOmega_{oo,g}^{-1}\bDelta_{o,g} )& \nonumber\\
                     &= \bm{m}_{c,g} + t_i \bm{\psi}_{c,g},\nonumber
         \end{align}
                and
          \begin{align}
           \bOmegag^* &= \ku_i\bOmega_{mm,g} - (\ku_i\bOmega_{mo,g})(\ku_i\bOmega_{oo,g})^{-1}(\ku_i\bOmega_{om,g})&\nonumber \\ 
                      &= \ku_i(\bOmega_{mm,g} - \bOmega_{mo,g} \bOmega_{oo,g}^{-1}\bOmega_{om,g}) \nonumber&\\
                      &= \ku_i\bOmega_{c,g}. \nonumber&
           \end{align}
   \end{proof}
\subsection{Parameter Estimation}
        This subsection examines the limitations of the typical log-likelihood optimisation procedure when applied to an MAR random sample. In particular, it identifies where missingness obstructs conventional parameter estimation and explicitly states the expressions that are unavailable, preventing further progress. In response, this subsection derives the necessary quantities to enable a viable model fitting methodology.
        
        In typical fashion, the parameters are estimated by optimising the log-likelihood function. 
        The following denote the collection of parameters $\allmu = (\bmu_1,\dots,\bmu_G)$,  $\allsig= (\bSig_1, \dots,\bSig_G)$, and $ \alllam = (\blam_1,\dots,\blam_G)$ and a complete sample $\bm{D} = (\bx_1, t_{1}, u_1, \bm{z}_{g,1}, \dots, \bx_n, t_{n}, u_n, \bm{z}_{g,n})$. Then from stochastic expression (\ref{sr}) the complete likelihood function $\mathcal{L}$ is derived as:
   \begin{align}
        \mathcal{L}\left(\allpi,\allmu, \allsig, \alllam, \alltheta|\bm{D}\right)
        &= \prod_{i=1}^n \prod_{g=1}^G \left[ \pi_g f(\bx_i,| t_i, u_i, z_{i,g})f(t_i| u_i, z_{i,g})f(u_i| z_{i,g})\right]^{z_{i,g}}& \nonumber\\
        &= \prod_{i=1}^n \prod_{g=1}^G \left[ \pi_g f_{\text{SN}}(\bx_i; \bmug + t_i\bDeltag, \ku_i\bOmegag, \bm{0})f_{\text{TN}}(t_i;0,\ku_i)h(u_i;\btheta_g)\right]^{z_{i,g}},&
    \end{align}
    from which the log-likelihood function $l_c$ is simplified to:
    \begin{align}
    \label{complete ll}
       l_c(\allpi, \allmu, \allsig, \alllam, \alltheta|\bm{D}) = \sum_{i=1}^n  \sum_{g=1}^G \bm{z}_{i,g} \left[ c_{i,g} + \ln(\pi_g) -\half \ln|\bOmegag| - \half \ku^{-1}_i \left(\bx_i - \bmug - \bDeltag t_i\right)^{\top} \bOmegag^{-1}\left(\bx_i - \bmug - \bDeltag t_i\right) \right]\hspace{-0.08cm},
   \end{align}
    where $c_{i,g} = \ln(\pi_g) - \frac{p}{2}\ln(2 \pi) + \half \ln(\frac{2}{\ku_i \pi}) - \half\frac{t_i^2}{k_i}+ \ln(h(u_i; \btheta_g))$. In a usual setting, $l_c$ is optimised via the parameters $\allpi, \allmu, \allsig, \alllam$ and $\alltheta$. However, the random variables $T, U$ and $Z_{g}$ are latent and $\bX$ is potentially incomplete. For this reason the complete log-likelihood $l_c$ is optimised by means of an Expectation-Maximisation (EM) based algorithm (\cite{EMmclachlan}). The EM algorithm is a popular iterative algorithm to maximise the log-likelihood for incomplete data. The expected value of the complete log-likelihood conditioned on the observed data, $\ev\left[l_c(\allpi, \allmu, \allsig, \alllam, \alltheta)| \underline{\bX^{o}} \right]$, say $ Q(\allpi, \allmu, \allsig, \alllam, \alltheta) $ is computed in the E-step and then maximised in the M-step in an iterative fashion until parameter estimates stabilise. The function $Q$ now depends on computing the following expected values for $i=1,\dots,n$ and $g = 1, \dots, G$:
   \begin{align}
   \label{ll cond ev md}
   &\widehat{z \ku^{-1} t\bx}_{i,g}        =\mathbb{E}\left[Z_{i,g}K_i^{-1}T_i\bX_i| \bX^{o}_i  \right],\hspace*{0.25cm}
   \widehat{z \ku^{-1} \bx }_{i,g}           =\mathbb{E}\left[Z_{i,g}K_i^{-1}\bX_i| \bX^{o}_i  \right], \hspace{0.25cm}
   \widehat{z \ku^{-1}\bx\bx^{\top}}_{i,g}   =\mathbb{E}\left[Z_{i,g}K_i^{-1}\bX_i\bX_i^{\top}| \bX^{o}_i\right],&
   \end{align}
   and
   \begin{align}
   \label{ll cond ev uni}
   &\widehat{z}_{i,g}                      = \mathbb{E}\left[Z_{i,g}|\bX^{o}_i\right], \hspace{0.25cm}
   \widehat{z\ku^{-1}}_{i,g}          = \mathbb{E}\left[Z_{i,g}K_i^{-1}| \bX^{o}_i  \right],\hspace*{0.25cm} 
   \widehat{z \ku^{-1} t }_{i,g}  = \mathbb{E}\left[Z_{i,g}K_i^{-1}T_i| \bX^{o}_i  \right],\hspace*{0.25cm}
   \widehat{z \ku^{-1} t^2}_{i,g}  = \mathbb{E}\left[Z_{i,g}K_i^{-1}T_i^2| \bX^{o}_i  \right].&
   \end{align}

Theorems \ref{sn cond dist} and \ref{norm cond dist} are now used to calculate the following set of new conditional expectations of $\bX_i$ in  \eqref{ll cond ev md}. These expressions are needed for the E-step of the new EM algorithm.
   \begin{theorem} 
   \label{sn cond ev thm}
   Suppose $\bX_i \in \mathbb{R}^p$ has a distribution with pdf (\ref{fmm}). Further, let $\dot{\blam}_{o,g} = \frac{\bSig_{oo,g}^{-1/2}\bDelta_{o,g} }{\sqrt{ 1 - \bDelta_{o,g}^{\top}\bSig_{oo,g}^{-1}\bDelta_{o,g} } }$, $ \bDelta_{c,g} =\frac{\bSig_{c,g}^{1/2}  \blam_{c,g}}{ \sqrt{1 +  \blam_{c,g}^{\top} \blam_{c,g} }}$, and $\xi_{i,g}^{\frac{r}{2}} =  \ku_i^{\frac{r}{2}}  W_{\phi}\left( \ku_i^{-\half} \dot{\blam}_{o,g}^{\top}\bSig_{oo,g}^{-1/2}(\bx_{i}^{(o)} - \bmu_{o,g}) \right) $. Then:\\
   \\
   \begin{tabular}{lcl}
        $\mathbb{E}\left[\bX_i| \bX^{o}_i, U_i, Z_{i,g} \right]$        & =     & $\begin{pmatrix} \bmu_{c,g} +\xi_{i,g}^{\half}\bDelta_{c,g} \\ \bx^{(o)}_i \end{pmatrix} $ \\
                                                                    and &       &  \\
        $\mathbb{E}\left[\bX_i| \bX^{o}_i, T_i, U_i, Z_{i,g} \right]$   & =     & $\begin{pmatrix} \bm{m}_{c,g} + t_i \bm{\psi}_{c,g}\\ \bx^{(o)}_i \end{pmatrix}.$
   \end{tabular}
    \begin{proof}
        Consider the missing and observed parts of $\bX_i$ as partitioned in (\ref{partitions}). Then \\
        \begin{align}\mathbb{E}\left[\bX_i| \bX^{o}_i, U_i, Z_{i,g} \right] = \begin{pmatrix} \vspace{0.2cm}\mathbb{E}\left[\bXm_i| \bXo_i, U_i, Z_{i,g}\right] \\
        \mathbb{E}\left[\bXo_i| \bXo_i, U_i, Z_{i,g} \right]\end{pmatrix} \text{ and }
        \mathbb{E}\left[\bX_i| \bX^{o}_i,T_i,  U_i, Z_{i,g} \right] = \begin{pmatrix} \vspace{0.2cm}\mathbb{E}\left[\bXm_i| \bXo_i,T_i, U_i, Z_{i,g}\right] \\
        \mathbb{E}\left[\bXo_i| \bXo_i, T_i, U_i, Z_{i,g} \right]\end{pmatrix}.
        \end{align}
        Since $\bXo_i$ is observed, $\mathbb{E}\left[\bXo_i| \bXo_i, U_i, Z_{i,g} \right] = \mathbb{E}\left[\bXo_i| \bXo_i, T_i, U_i, Z_{i,g} \right]=\bx^{o}_i$. Substituting the parameters from Theorem \ref{sn cond dist} and Theorem \ref{norm cond dist} into the first result of  Theorem \ref{sm moments} yield:\\
     \begin{align}
     \label{sn cond ev}
         \mathbb{E}\left[\bX_i^{(m)} | \bX^{(o)}_i, U_i, Z_{i,g} \right] & = \bmu_{c,g} + \bDelta_{c,g}\ku_i^{\half} W_{\phi}\left( \ku_i^{\half}\delta_{0c,g} \right), \text{ and }& \\
         \label{norm cond ev}
         \mathbb{E}\left[\bX_i^{(m)} | \bX^{(o)}_i, T_i, U_i, Z_{i,g} \right] & = \bmu_{c,g} +t_i\bm{\psi}_{c,g}, \text{ where }&
     \end{align}
     $\delta_{0c,g} = \frac{ \lambda_{0,c,g} }{ \sqrt{1 +  \blam_{c,g}^{\top} \blam_{c,g} }  }$. Finally, substituting the result from Lemma \ref{lemma2} in the Appendix into the expected value (\ref{sn cond ev}) yields:
     \begin{align}
         \mathbb{E}\left[\bX_i^{(m)} | \bX^{(o)}_i, U_i, Z_{i,g} \right] & = \bmu_{c,g} +\bDelta_{c,g} \ku_i^{\half} W_{\phi}\left( \ku_i^{-\half} \dot{\blam}_{o,g}^{\top}\bSig_{oo,g}^{-1/2}(\bx_{i}^{(o)} - \bmu_{o,g}) \right) =  \bmu_{c,g} +\xi_{i,g}^{\half}\bDelta_{c,g}.&
     \end{align}
       \end{proof}
   \end{theorem}
Lastly, the conditional expected value of the cross-product is derived next. Following the partition (\ref{partitions}):
    \begin{align}
    \vspace{0.75cm}
        \bX_i\bX_i^{\top} &= \begin{pmatrix} \bX_i^{(m)}\bX_i^{(m)\top}  & \bX_i^{(m)}\bX_i^{(o)\top} \\
                                            \bX_i^{(o)}\bX_i^{(m)\top}  & \bX_i^{(o)}\bX_i^{(o)\top} 
                            \end{pmatrix} \hspace{0.2cm} \text{which implies that,} &\nonumber\\
          \vspace{0.75cm}
          \mathbb{E}\left[\bX_i\bX_i^{\top}| \bX^{o}_i, U_i, Z_{i,g} \right] &= 
          \begin{pmatrix} \ev[\bX_i^{(m)}\bX_i^{(m)\top}| \bXo_i, U_i, Z_{i,g}]  & \ev[\bX_i^{(m)}| \bXo_i, U_i, Z_{i,g}]\bx_i^{(o)\top} \\
                          \bx_i^{(o)} \ev[\bX_i^{(m)}| \bXo_i, U_i, Z_{i,g}]^{\top}  & \bx_i^{(o)}\bx_i^{(o)\top} 
          \end{pmatrix}.&
    \end{align}
    The conditional expected value $\ev[\bX_i^{(m)\top}| \bXo_i, U_i, Z_{i,g}]$ was derived in Theorem \ref{sn cond dist}. All that is left to derive is the conditional expected value of the cross-product between the missing components of $\bX_i^{(m)}$ - which leads to the following theorem.
\begin{theorem}
    \label{sn cond cross ev thm}
    Suppose $\bX_i \in \mathbb{R}^p$ has the pdf (\ref{fmm}). Consider the missing and observed parts of $\bX_i$ as partitioned in (\ref{partitions}). Then
        \begin{align}
        \ev[\bX_i^{(m)}\bX_i^{(m)\top}| \bX^{o}_i, U_i, Z_{i,g}] = \ku_i\bSig_{c,g} + \bmu_{c,g}\bmu_{c,g}^{\top} +\left(\bmu_{c,g}\bDelta_{c,g}^{\top} + \bDelta_{c,g}\bmu_{c,g}^{\top} \right)\xi_{i,g}^{\half}  - A^{(o)}_{i,g} \bDelta_{c,g}\bDelta_{c,g}^{\top} \xi_{i,g}^{\half}, 
        \end{align}
        where $\xi_{i,g}^{\half}$ is defined in  Theorem \ref{sn cond ev thm} and $A_{i,g}^{(o)} = \dot{\blam}_{o,g}^{\top}\bSig_{oo,g}^{-1/2}(\bx_{i}^{(o)} - \bmu_{o,g})$.
        \begin{proof}
        The covariance matrix of $ \bX_i$ conditioned on $\bXo_i$, $K_i$ and $Z_{i,g}$ will be simplified. Substituting the parameters from Theorem \ref{sn cond dist} into the second result of Theorem \ref{sm moments} and letting $A_{i,g}^{(o)} = \dot{\blam}_{o,g}^{\top}\bSig_{oo,g}^{-1/2}(\bx_{i}^{(o)} - \bmu_{o,g})$ yields:
        \begin{align}
        \label{cov cond}
            \mathrm{Cov}[\bX_i|\bX^{o}_i, U_i, Z_{i,g}] 
            &= \ku_i\bSig_{c,g} - \bDelta_{c,g}\bDelta_{c,g}^{\top}( \delta_{0c,g}\ku_i^{\half} W_{\phi}(\ku_i^{-\half} \delta_{0c,g}) + W_{\phi}^2(\ku_i^{-\half} \delta_{0c,g}) ) &\nonumber\\
            &= \ku_i\bSig_{c,g} - A_{i,g}^{(o)}\bDelta_{c,g}\bDelta_{c,g}^{\top} \xi_{i,g}^{\half} - \xi_{i,g}^{\half}\bDelta_{c,g}\bDelta_{c,g}^{\top}\xi_{i,g}^{\half}.
        \end{align}
        The following cross-product is simplified. Using the first result from Theorem \ref{sn cond ev thm}:
        \begin{align}
            \ev[\bX_i^{(m)}| \bX^{o}_i, U_i, Z_{i,g}]  \ev[\bX_i^{(m)}| \bX^{o}_i, U_i, Z_{i,g}] ^{\top} 
            & = \left( \bmu_{c,g} + \xi^{\half}_{i,g}\bDelta_{c,g} \right)\left( \bmu_{c,g} + \xi^{\half}_{i,g}\bDelta_{c,g} \right)^{\top} & \nonumber\\
            & = \bmu_{c,g}\bmu_{c,g}^{\top} + \xi^{\half}_{i,g}\left(\bmu_{c,g}\bDelta_{c,g}^{\top} + \bDelta_{c,g}\bmu_{c,g}^{\top} \right) +  \xi^{\half}_{i,g}\bDelta_{c,g} \bDelta_{c,g}^{\top} \xi^{\half}_{i,g}.&
        \end{align}
        The conditional expected value of $\bX_i^{(m)}\bX_i^{(m)\top}$ is calculated as follows:
           \begin{align}
               \ev[\bX_i^{(m)}\bX_i^{(m)\top}| \bX^{o}_i, U_i, Z_{i,g}] 
               &=\mathrm{Cov}[\bX_i|\bX^{o}_i, U_i, Z_{i,g}]  + \ev[\bX_i^{(m)}| \bX^{o}_i, U_i, Z_{i,g}]  \ev[\bX_i^{(m)}| \bX^{o}_i, U_i, Z_{i,g}] ^{\top},& \nonumber \\
               &= \ku_i\bSig_{c,g} + \bmu_{c,g}\bmu_{c,g}^{\top} + \left( \bmu_{c,g}\bDelta_{c,g}^{\top} + \bDelta_{c,g}\bmu_{c,g}^{\top} \right)\xi_{i,g}^{\half}  - A^{(o)}_{i,g} \bDelta_{c,g}\bDelta_{c,g}^{\top} \xi_{i,g}^{\half}.& \nonumber
           \end{align}
        \end{proof}
\end{theorem}
The conditional expected values (\ref{ll cond ev uni}) will be derived next so that all necessary expectations are obtained for the E-step. We start off by deriving $\ev[Z_{i,g}=1| \bXo_i]$:
   \begin{align}
   \label{z ev}
    \widehat{z}_{i,g} &= \mathbb{P}[Z_{i,g}=1| \bXo_i] = \frac{\pi_g f_{\text{SMSN}}(\bx^{(o)}_i; \bmu_{o,g}, \bSig_{oo.g},\dot{\blam}_{o,g},\btheta_g)}{ \sum_{g=1}^G \pi_g f_{\text{SMSN}}(\bx^{(o)}_i;\bmu_{o,g}, \bSig_{oo.g},\dot{\blam}_{o,g},\btheta_g)}.&
\end{align}
Bayes theorem is used to determine that, $T_i|\bX^{o}_i = \bx^{o}_i, U_i = u_i ,Z_{i,g} = 1 \sim TN(\mu_{T_{i,g}}, \ku_i\sigma^2_{T_{g}})$,  where $\sigma^2_{T_{g}} = \left(1 + \bDelta_{o,g}^{\top}\bOmega_{oo,g}^{-1}\bDelta_{o,g} \right)^{-1}$ and $ \mu_{T_{i,g}} = \sigma^2_{T_{g}} \bDelta_{o,g}^{\top}\bOmega_{oo,g}^{-1}(\bx^{(o)}_i - \bmu_{o,g})$. The detailed derivation is provided in the Appendix in Theorem \ref{t cond dist}. From Definition $\ref{trunc normal}$, and noting from Lemma \ref{lemma3} in the Appendix that $\frac{\mu_{T_{i,g}}}{\sigma_{T_{g}}} = A_{i,g}^{(o)} $ we get the following expected values:
\begin{align}
    \mathbb{E}\left[T_i| \bX^{o}_i, U_i,Z_{i,g}\right]    &= \mu_{T_{i,g}} + \sigma_{T_{g}}\ku_i^{\half} W_{\phi}\left(A_{i,g}^{(o)}\right) =  \mu_{T_{i,g}} + \sigma_{T_{g}} \xi_{i,g}^{\frac{1}{2}}& \\
     \mathbb{E}\left[T_i^2| \bX^{o}_i, U_i,Z_{i,g}\right] &= \mu^2_{T_{i,g}} + \mu_{T_{i,g}}\sigma_{T_{g}}  \ku_i^{\half}W_{\phi}\left(A_{i,g}^{(o)}\right) + \ku_i\sigma^2_{T_{g}} = \mu^2_{T_{i,g}} + \mu_{T_{i,g}}\sigma_{T_{g}}  \xi_{i,g}^{\frac{1}{2}} + \ku_i\sigma^2_{T_{g}} .&
\end{align}
The rest of the conditional expected values have the following expressions:
\begin{align}
    \widehat{z \ku^{-1} }_{i,g}     & = \widehat{z}_{i,g}\widehat{\ku^{-1}}_{i,g} &\nonumber\\
    \widehat{z \ku^{-1} t }_{i,g}   & = \widehat{z}_{i,g} \left( \widehat{\ku^{-1} }_{i,g} \mu_{T_{i,g}} + \sigma_{T_{g}}  \widehat{\xi_{i,g}^{-\half}} \right)                            =\widehat{z}_{i,g}\widehat{\ku^{-1}t}_{i,g}&\nonumber\\
    \widehat{z \ku^{-1} t^2}_{i,g}  & = \widehat{z}_{i,g} \left( \widehat{\ku^{-1} }_{i,g} \mu^2_{T_{i,g}} +\mu_{T_{i,g}}\sigma_{T_{g}}\widehat{\xi_{i,g}^{-\half}} + \sigma^2_{T_{g}} \right) =\widehat{z}_{i,g}\widehat{\ku^{-1} t^2}_{i,g}, \hspace{0.2cm} \text{ where }\nonumber &
\end{align}
 $\widehat{ \ku^{-1}t }_{i,g} = \widehat{\ku^{-1} }_{i,g} \mu_{T_{i,g}} + \sigma_{T_{g}}  \widehat{\xi_{i,g}^{-\half}} $ and $\widehat{\ku^{-1} t^2 }_{i,g} =\widehat{\ku^{-1} }_{i,g} \mu^2_{T_{i,g}} +\mu_{T_{i,g}}\sigma_{T_{g}}\widehat{\xi_{i,g}^{-\half}} + \sigma^2_{T_{g}} $. With the help of the conditional expectations in Theorem \ref{sn cond ev thm} and Theorem \ref{sn cond cross ev thm}, the conditional expected values (\ref{ll cond ev md}) have the following expressions:
\begin{align}
    \widehat{z \ku^{-1} \bx}_{i,g}           &= \widehat{z}_{i,g} \begin{pmatrix} \widehat{\ku^{-1} \bx}_{i,g}^{(m)} \\ \bx^{(o)}_i \end{pmatrix}, &\nonumber\\
    \widehat{z\ku^{-1} t \bx }_{i,g}         &= \widehat{z}_{i,g} \begin{pmatrix} \widehat{\ku^{-1} t }_{i,g} \bm{m}_{c,g} + \widehat{\ku^{-1} t^2}_{i,g} \bm{\psi}_{c,g}\\ \bx^{(o)}_i \end{pmatrix},  &\nonumber\\
    \widehat{z\ku^{-1}\bx\bx^{\top}}_{i,g}   &= \widehat{z}_{i,g} \begin{pmatrix} 
                                                \bSig_{c,g} + \widehat{\ku^{-1}}_{i,g}\bmu_{c,g}\bmu_{c,g}^{\top} +\bm{\alpha}_{i,g} \xi_{i,g}^{-\half}  & \widehat{\ku^{-1} \bx}_{i,g}^{(m)}\bx_i^{(o)\top} \\
                                                \bx_i^{(o)}\widehat{\ku^{-1} \bx}_{i,g}^{(m)\top}  & \bx_i^{(o)}\bx_i^{(o)\top} \end{pmatrix}, \hspace{0.2cm} \text{ where }& \nonumber
   \end{align}
   $\widehat{\ku^{-1} \bx}_{i,g} ^{(m)} = \widehat{\ku^{-1}}_{i,g}\bmu_{c,g} +\widehat{\xi_{i,g}^{-\half}}\bDelta_{c,g}$,\hspace{0.2cm} and \hspace{0.2cm} $ \bm{\alpha}_{i,g} = \bmu_{c,g}\bDelta_{c,g}^{\top} + \bDelta_{c,g}\bmu_{c,g}^{\top} - A^{(o)}_{i,g} \bDelta_{c,g}\bDelta_{c,g}^{\top}$.

\subsection{EM based algorithm }
%
With the required expected values now derived, we can illustrate how the EM algorithm is applied. Let $\bm{P}^{(k)} = (\widehat{\allpi}^{(k)}, \widehat{\allmu}^{(k)}, \widehat{\allsig}^{(k)}, \widehat{\alllam}^{(k)}, \widehat{\alltheta}^{(k)})$ denote the parameter estimates at the $k^{\text{th}}$ iteration. Define the conditional expectations $\widehat{z \ku^{-1} \bx}_{i,g}^{(k)}$,  $\widehat{z\ku^{-1} t\bx}_{i,g} ^{(k)}$,  $\widehat{z\ku^{-1}\bx\bx^{\top}}_{i,g}^{(k)}$,  $\widehat{z\ku^{-1} }_{i,g}^{(k)}$,  $\widehat{z\ku^{-1} t^2}_{i,g}^{(k)}$,  and $\widehat{z\ku^{-1} t }_{i,g}^{(k)}$ as the current values computed using $\bm{P}^{(k)}$.
In the E-step, we compute the expected complete-data log-likelihood (\ref{complete ll}) under $\bm{P}^{(k)}$, denoted as $Q(\bm{P}|\bm{P}^{(k)})$. Differentiating this function does not yield closed-form solutions for each parameter independently, but it does result in simpler expressions that are interdependent. Consequently, the M-step is divided into several conditional maximisation (CM) steps, where each parameter is updated in turn while the others are held fixed. This forms the Expectation Conditional Maximisation (ECM) algorithm. Once $\bm{P}^{(k+1)}$ is obtained from the CM steps, the process repeats until convergence. The ECM procedure for the FMSMSN family is detailed in Algorithm \ref{em alg}.
\newline
The EM-type algorithm proposed in Algorithm \ref{em alg} requires initial values that are close enough to the population parameters to ensure it converges to the global maximum. One strategy is to simulate random parameters or randomly partition the data and run the M-step of the ECM algorithm. However, inadequate starting values may cause the algorithm to converge to a local maximum. More informed initialisation techniques in the literature surveyed operate on the idea of using other clustering techniques available to partition data as initialisation, of which a comprehensive review of these techniques can be found in \cite{EMinitialisation} and \cite{melnykov2012initializing}. Some examples include clustering via hierarchical clustering and k-means clustering. The strategy used in this paper considers clusters partitioned by the k-means algorithm as a suitable starting point.
\newline
The log-likelihood of the observed dataset increases monotonically at each iteration of the ECM algorithm. However, the algorithm may run into a local maximum before it finds a global maximum, in which case the initial stability is temporary -i.e. the the observed log-likelihood values are stable only for a few iterations before they increase again. The Aitken acceleration criterion is thus used in this paper to determine whether the algorithm has converged to its asymptotic value. Let $l_o^{(k)}$ denote the observed log-likelihood at the $k^{th}$ iteration. Then the Aitken acceleration criterion is given as:
\begin{align}
    a^{(k+1)} = \frac{l_o^{(k+2)} - l_o^{(k+1)}}{l_o^{(k+1)} - l_o^{(k)}}.
\end{align}
Then the estimated asymptotic observed log-likelihood at the $k^{th}$ iteration, say $(l_o^{\infty})^{(k)}$ is:
\begin{align}
    (l_o^{\infty})^{(k)} = l_o^{(k+1)} + \frac{ l_o^{(k+2)} - l_o^{(k+1)}  }{1 - a^{(k+1)}}.
\end{align}
The EM algorithm is therefore considered to have converged if $(l_o^{\infty})^{(k)} - l_o^{(k+1)} < \epsilon$ where $\epsilon>0$ is a small number.
\begin{algorithm}
\caption{ECM for FMSMSN with MAR}\label{em alg}
\KwIn{Initial parameter estimates $\bm{P}^{(k)}$}
\KwOut{Updated parameter estimates $\bm{P}^{(k+1)}$}

\textbf{E-step:} \\
\Indp Compute $Q(\bm{P}|\bm{P}^{(k)})$ - that is, compute conditional expectations (\ref{ll cond ev md}) and (\ref{ll cond ev uni}) using $\bm{P}^{(k)}$.\\
\Indm

\textbf{M-step:}\\
\Indp Differentiating $Q(\bm{P}|\bm{P}^{(k)})$ with respect to $\bm{P}$ and update $\bm{P}^{(k+1)}$ as follows:\\
\For{$g = 1, \dots, G$}{
     $\hat{\pi}_g^{(k+1)} = \frac{1}{n} \displaystyle\sum_{i=1}^{n}\widehat{z}_{i,g}^{(k)}$\newline
     \vspace{0.2cm} \newline
    $\begin{aligned}[t]
       \widehat{\bOmega}_g^{(k+1)} = \frac{1}{ \displaystyle\sum_{i=1}^{n}\widehat{z}_{i,g}^{(k)}}  \displaystyle\sum_{i=1}^{n} &\Big\{\widehat{z\ku^{-1}\bx\bx^{\top}}_{i,g} - \widehat{z \ku^{-1} \bx}_{i,g}\hat{\bmug}^{(k)\top} - \hat{\bmug}^{(k)}\widehat{z \ku^{-1} \bx}_{i,g}^{\top} + \widehat{z \ku^{-1}}_{i,g}^{(k)}\hat{\bmug}^{(k)} \hat{\bmug}^{(k) \top} - \widehat{z\ku^{-1} t \bx }_{i,g}(\bDeltag^{(k)})^{\top} \\
        &  - \bDeltag^{(k)}\widehat{z\ku^{-1} t \bx }_{i,g}^{\top} + \widehat{z \ku^{-1} t }_{i,g}\Big(\hat{\bmug}^{(k+1)} (\bDeltag^{(k)})^{\top} + \bDeltag^{(k)} (\hat{\bmug}^{(k+1)})^{\top}\Big) + \widehat{z \ku^{-1} t^2}_{i,g}\bDeltag^{(k)}(\bDeltag^{(k)})^{\top}\Big\}
    \end{aligned}$\newline
    \vspace{0.2cm} \newline
    $\widehat{\bmu}_g^{(k+1)} = \frac{1}{ \displaystyle\sum_{i=1}^{n} \widehat{z\ku^{-1}}_{i,g}^{(k)}}   \displaystyle\sum_{i=1}^{n}\left\{ \widehat{z\ku^{-1} \bx}_{i,g} ^{(k)}- \widehat{z \ku^{-1}t}_{i,g} ^{(k)}\bDeltag^{(k)}  \right\} $\newline
    \vspace{0.2cm} \newline
    $\widehat{\bDelta}_g^{(k+1)} = \frac{1}{ \displaystyle\sum_{i=1}^{n} \widehat{z \ku^{-1} t^2}_{i,g}^{(k)} }   \displaystyle\sum_{i=1}^{n}\left\{\widehat{z\ku^{-1} t \bx }_{i,g} - \widehat{z \ku^{-1} t}_{i,g}\hat{\bmug}^{(k+1)}  \right\}$\newline
    \vspace{0.2cm} \newline
    $\widehat{\btheta}_g^{(k+1)} =   \underset{\btheta_g }{\arg\max} \{ Q(\bm{P}|(\hat{\pi}_g^{(k+1)}, \widehat{\bOmega}_g^{(k+1)}, \widehat{\bmu}_g^{(k+1)},\widehat{\bDelta}_g^{(k+1)})  )\} $\newline
    \vspace{0.2cm} \newline
    $\widehat{\bSig}_g^{(k+1)}  =  \widehat{\bOmega}_g^{(k+1)} + \widehat{\bDelta}_g^{(k+1)}(\widehat{\bDelta}_g^{(k+1)})^{\top}$, and \newline
    \vspace{0.2cm} \newline
    $\widehat{\blam}_g^{(k+1)} = \frac{ (\widehat{\bSig}_g^{(k+1)} )^{-1/2} \widehat{\bDelta}_g^{(k+1)} }{\sqrt{1 - (\widehat{\bDelta}_g^{(k+1)} )^{\top} (\widehat{\bSig}_g^{(k+1)} )^{-1} \widehat{\bDelta}_g^{(k+1)} } }$
}
\Indm
iterate E and M steps until $(l_o^{\infty})^{(k)} - l_o^{(k+1)} < \epsilon$ where $\epsilon>0$ is a small number.
\end{algorithm}

\section{Simulation experiments}
\label{sec:simulations}
This section describes the simulation design and performance metrics, and compares the results of fitting the multivariate skew-normal, multivariate skew-t, skew-slash, and skew-variance-gamma distributions on simulated data. Section \ref{experimentdesign} explains the simulation design, Section \ref{compconsiderations} discusses the computational pitfalls that may be encountered, followed by the clustering and parameter recovery performances in Section \ref{sim results}.
\subsection{Experiment design}
\label{experimentdesign}
Throughout the experiment, bivariate data from a two-component mixture are simulated with the following scales matrices and skewness vectors: \newline
\begin{align}
   \pi_1 = 0.3, \hspace{0.35cm} \pi_2 = 0.7, \hspace{0.35cm} \blam_1 = \begin{pmatrix} 3 \\ 6 \end{pmatrix}, \hspace{0.35cm} \blam_2 = \begin{pmatrix} 5 \\ 4 \end{pmatrix}, \hspace{0.35cm} \bmu_1 = \begin{pmatrix} -5 \\ 0 \end{pmatrix}, \hspace{0.35cm} \bSig_1 = \begin{pmatrix}3 & -1 \\ -1 & 3\end{pmatrix} , \hspace{0.35cm} \text{ and } \hspace{0.35cm} \bSig_2 = \begin{pmatrix}3 & 1 \\ 1  & 3\end{pmatrix}.
\end{align}
We account for varying levels of overlap between the two clusters. We consider well separated and close clusters when $ \bmu_2 = \begin{pmatrix} -3 \\ 0 \end{pmatrix}$ and $ \bmu_2 =\begin{pmatrix} -1 \\ 0 \end{pmatrix} $, respectively. For both of these cases we consider varying proportions of values missing from the sample: specifically, we consider samples where 0\%, 20\%, 40\%, 60\%, and 80\% of the observed vectors have values missing at random. For each cluster overlap and missingness proportion scenario, we simulate $B =200$ samples of size $n=200$ from the following distributions:
\begin{enumerate}
    \item Multivariate skew-normal.
    \item Multivariate skew-slash where $\alpha_1 = 3$ and $\alpha_2 =2$.
    \item Multivariate skew-t where $\nu_1 = 4$ and $\nu_2 = 7$.
    \item Multivariate skew-variance-gamma, where $\eta_1 = 2$ and $\eta_2 = 3$.
\end{enumerate}
For each of the data generating processes, all four distributions (multivariate skew-normal, skew-slash, skew-t, and skew-variance-gamma) are fitted. The experiments are assessed according to the clustering performance of the distributions and the parameter recovery. To assess the clustering performances of the distributions, the average Adjusted Rand Index (ARI) values are recorded. The ARI is a measure between 0 and 1 that assesses the likelihood that the clustering method's results are due to random chance. 
The ARI considers the agreements between the method's clusters and the true classes as a ratio of the total number of ways the data points could be partitioned  while accounting for the expected value of agreements due to random chance. An ARI value close to 0 indicates that the clustering method performs no better than random chance, while a value close to 1 indicates the clustering method has performed well.

To evaluate how well the fitted distributions recover the true parameters, we use the Absolute Bias (AB) and Root Mean Squared Error (RMSE) metrics, as defined in \cite{test}. 
Let $\btheta = (\btheta_1,\dots,\btheta_G)$ denote the true parameter values, where each $\btheta_g$ contains $p$ elements. For the $b^{\text{th}}$ replicate, let the estimated parameters be $\hat{\btheta}_b = (\hat{\btheta}_{1,b},\dots,\hat{\btheta}_{G,b})$, for $b = 1, \dots, B$. The AB and RMSE for group $g$ are defined as:
\begin{align}
\text{AB}(\btheta_g) &= \frac{1}{B} \sum_{b=1}^B \sum_{i=1}^p |\hat{\theta}_{b,i,g} - \theta_{i,g}| \hspace{0.5cm} \text{ and } \hspace{0.5cm} \text{RMSE}(\btheta_{g}) = \sqrt{ \frac{1}{B} \sum_{b=1}^B \sum_{i=1}^p (\hat{\theta}_{b,i,g} - \theta_{i,g})^2 }.
\end{align}
These measures are applied to the location parameters, the skewness parameters, the trace and anti-trace of the scale matrices in the simulation study to assess the parameter recovery performance.

\subsection{Computation considerations}
\label{compconsiderations}
The experiments attempt to fit the distributions in Table \ref{4cases} to simulated data generated from said distributions in a permutative manner. As discussed in Section \ref{special cases}, this causes the model fitting to follow a limiting case. For example, fitting a multivariate skew-t distribution on data generated from a multivariate skew-normal distribution causes its estimated degrees of freedom $\nu$ to increase without bound. Similarly, fitting a skew-slash distribution causes the estimated $\alpha$ to increase without bound. In these cases, the algorithm may crash as the numerical values become too large or too small to be handled by the programming language. Mitigating these problematic cases requires an upper bound on the values the estimated hyperparameters can achieve. On a similar note, the order of operations with respect to multiplication and division may produce values that fall outside the machine's precision. This effectively produces zeros in denominator positions, thereby interrupting the algorithm. Having precautionary checks that inspect whether or not the denominator values are within the machine's precision allows the algorithm to run smoothly. 

Estimating the hyperparameters required approximating $\ev[\ln(\ku_i^{-1})| \bX^{o}_i,Z_{i,g}]$, which was accomplished by using a second order Taylor polynomial expansion about $\ev[\ku_i^{-1}]$, which produces the following numerically stable expression:
\begin{align*}
    \ln(\ku_i^{-1})      & \approx \ln\left( \ev[\ku_i^{-1}| \bX^{o}_i,Z_{i,g}] \right) + \frac{ \ku_i^{-1} - \ev[\ku_i^{-1}| \bX^{o}_i,Z_{i,g}]  }{\ev[\ku_i^{-1}| \bX^{o}_i,Z_{i,g}] } - \frac{ (\ku_i^{-1} - \ev[\ku_i^{-1}| \bX^{o}_i,Z_{i,g}] )^2  }{2(\ev[\ku_i^{-1}| \bX^{o}_i,Z_{i,g}])^2 }, \text{ so that}\\
    \ev[\ln(\ku_i^{-1})| \bX^{o}_i,Z_{i,g}] & \approx \ln\left( \ev[\ku_i^{-1}| \bX^{o}_i,Z_{i,g}] \right) - \half \left(\frac{\ev[\ku_i^{-2}| \bX^{o}_i,Z_{i,g}]   }{\left( \ev[\ku_i^{-1}| \bX^{o}_i,Z_{i,g}] \right)^2} - 1 \right).
\end{align*}
This approximation is improved by realising that $\ln(\cdot)$ is a concave function on its domain. It is then possible to employ Jensen's inequality which constrains the values so that :
\begin{align}
    \frac{ \ev[\ku_i^{-2}| \bX^{o}_i,Z_{i,g}]   }{\left( \ev[\ku_i^{-1}| \bX^{o}_i,Z_{i,g}] \right)^2} >1.
\end{align}

Lastly, the algorithm that fits the multivariate skew-variance-gamma distribution has the tendency to estimate $\eta$ close to zero. It is discussed in detail by \cite{svgarticle} that the hyperparameter $\eta \leq \frac{p}{2}$ produces an infinite log-likelihood value. Some alterations have been suggested to work around these cases, but they do not preserve the monotonicity of the observed log-likelihood over each iteration. Similarly to handling the edge cases presented by multivariate skew-t and multivariate skew-slash while still preserving the monotonicity of the observed log-likelihood at each iteration, it is suggested that $\eta$ be bounded below by $\frac{p}{2}$.
\subsection{Simulation results}
\label{sim results}
The design of the simulation experiment involves a combination of several factors. First, two levels of cluster overlap are considered to reflect varying degrees of separation between mixture components. For each overlap level, five different proportions of missing data are introduced to assess performance under increasing levels of data incompleteness. These settings are applied across four distinct underlying mixture distributions, each representing a different data-generating processes. For each scenario, each of the four different distributions from \tablename  \ref{4cases} are fitted. This results in a total of 80 distinct simulation settings. The average Adjusted Rand Index (ARI) values obtained for each of these scenarios are presented in Section \ref{Cluster performance} and the parameter recovery performance is discussed in Section \ref{parameter recovery}.

\subsubsection{Clustering performance}
\label{Cluster performance}
\begin{figure}[H]
    \centering
    \includegraphics[width=0.67\linewidth, height = 12cm]{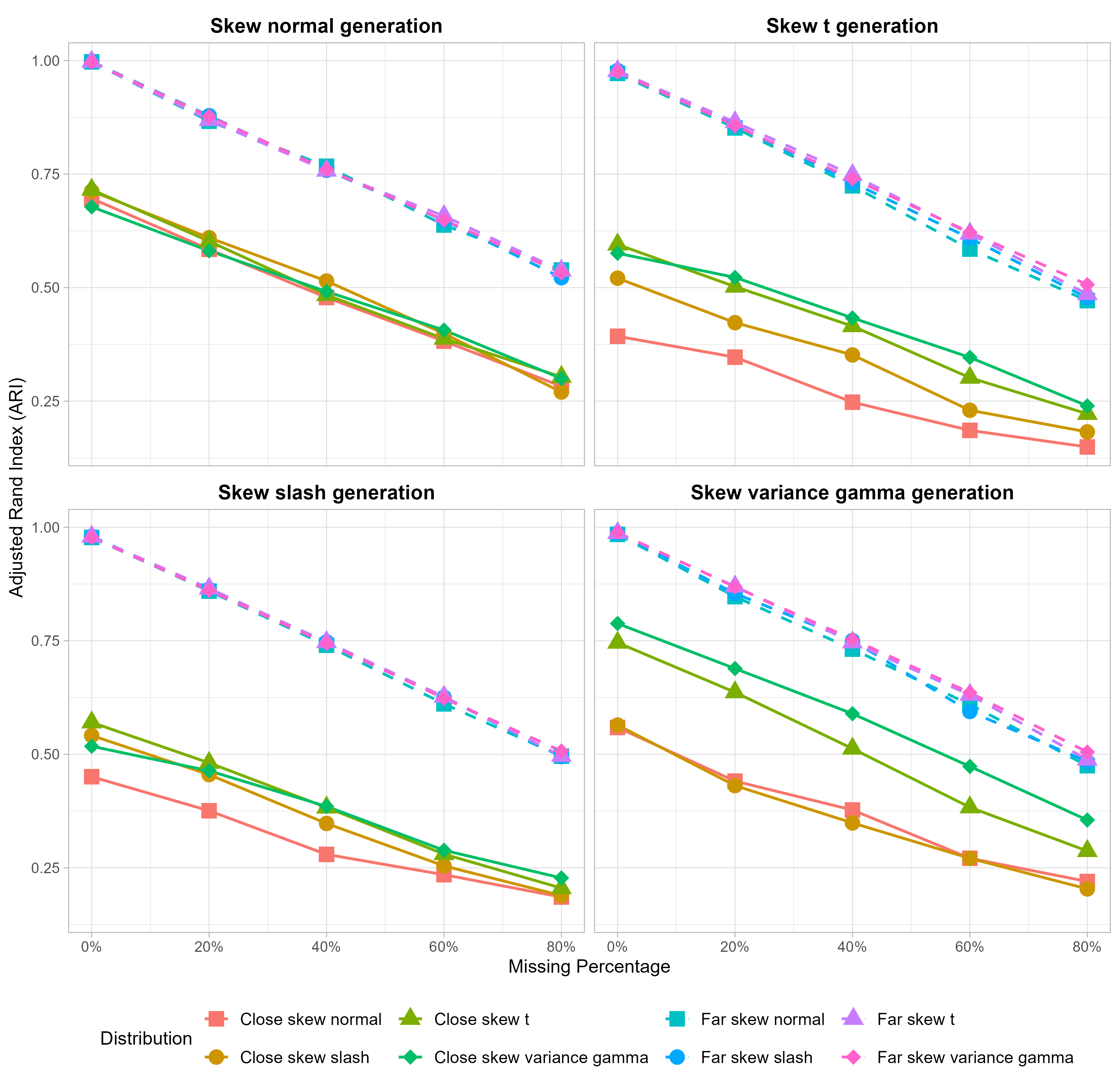}
    \caption{Average ARI values across 200 replications for datasets of size $n = 200$, randomly generated from a two-component mixture whose component distributions are specified in the headers of each subplot.}
    \label{ARI200}
\end{figure}

\begin{figure}[H]
    \centering\includegraphics[width=0.67\linewidth, height = 12cm]{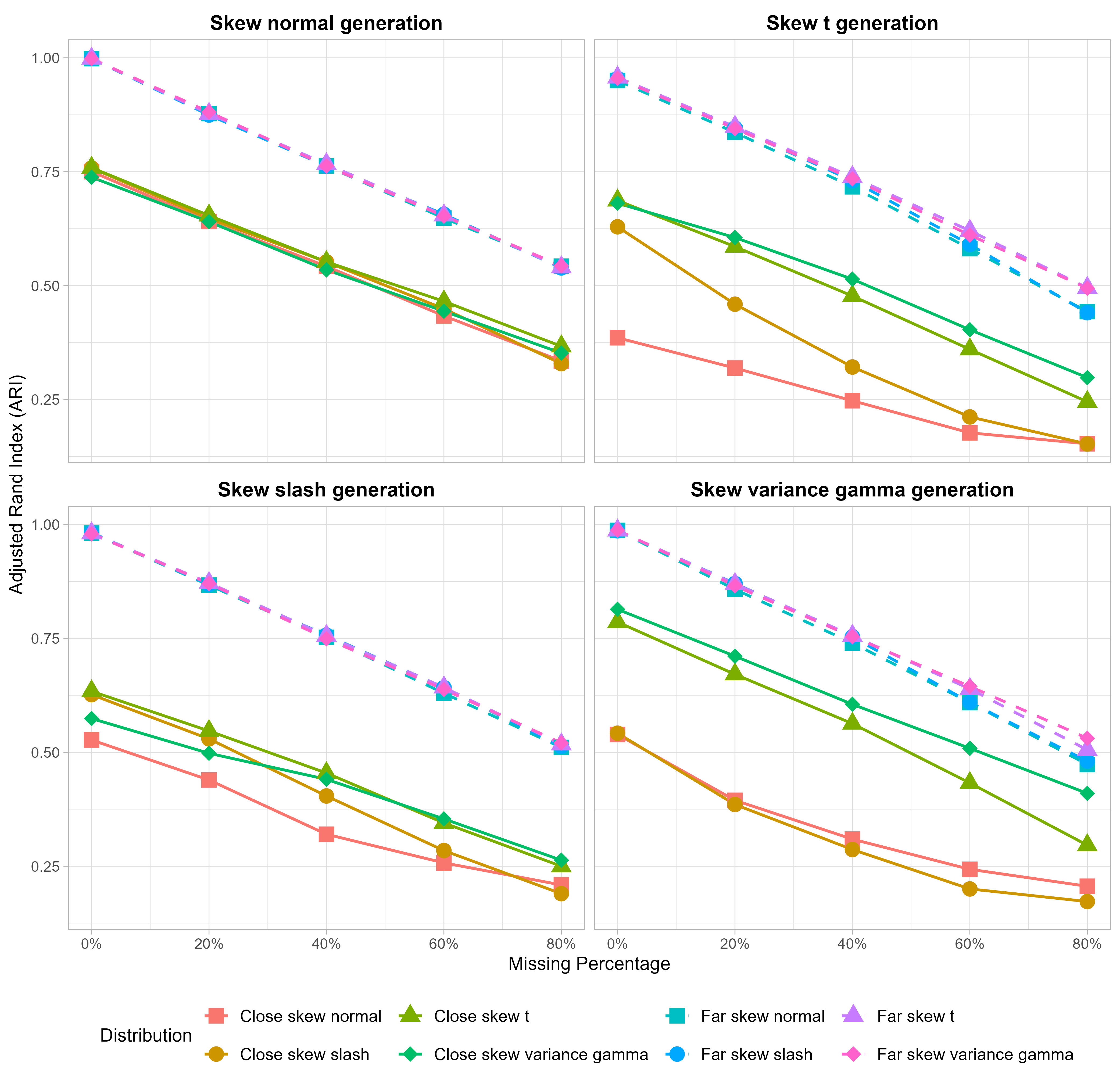}
    \caption{Average ARI values across 200 replications for datasets of size $n = 500$, randomly generated from a two-component mixture whose component distributions are specified in the headers of each subplot.}
    \label{ARI500}
\end{figure}

Overall, the patterns from Figures \ref{ARI200} and \ref{ARI500} indicate a clear inverse relationship between the proportion of incomplete data and clustering performance across all distributions: as the percentage of missing rows increases, clustering accuracy declines, regardless of cluster proximity or sample size. This behaviour is expected from previously proposed extensions of scale mixtures of multivariate normal distributions with MAR values (\cite{tt1, tt2}). A comparison of solid and dashed lines reveals that cluster proximity impacts the inherent clustering ability of each distribution, with closer proximity consistently degrading performance. However, proximity does not appear to intensify the effect of increasing missingness on clustering outcomes. The near-linear decline in ARI values is observed consistently across all data generation–distribution pairs when $n=200$.

In \figurename~\ref{ARI500}, the ARI deteriorates more rapidly for the multivariate skew-normal and multivariate skew-slash distribution applied to data generated by multivariate skew-t and multivariate skew-variance-gamma distributions—the latter being the most complex of the four special cases considered. All four distributions yield comparable performance when fit to data generated from a multivariate skew-normal distribution, which aligns with expectations since the multivariate skew-normal is a nested case within the multivariate skew-slash, skew-t, and skew-variance-gamma families.

A consistent performance hierarchy emerges: both the multivariate skew-t and skew-variance-gamma distributions approximate the performance of the skew-slash model when fitted to skew-slash-generated data, while the skew-variance-gamma model performs on par, if not better, than the multivariate skew-t distribution when fit to multivariate skew-t-generated data. This ordering reflects the increasing flexibility and complexity of the families, ranging from multivariate skew-normal to skew-variance-gamma. Notably, these performance gaps become more pronounced with closer cluster proximity and higher missing data percentages and are consistent across both sample sizes ($n=200$ and $n=500$). In contrast, for datasets with well-separated clusters, the differences in clustering performance among the distributions are much less stark.

Next, we look at the parameter recovery for each of the four distributions when fitted on their own generated data.

\subsubsection{Parameter recovery}
\label{parameter recovery}

\begin{figure}[H]
    \centering
    \begin{subfigure}{0.49\textwidth}
          \includegraphics[width=\textwidth]{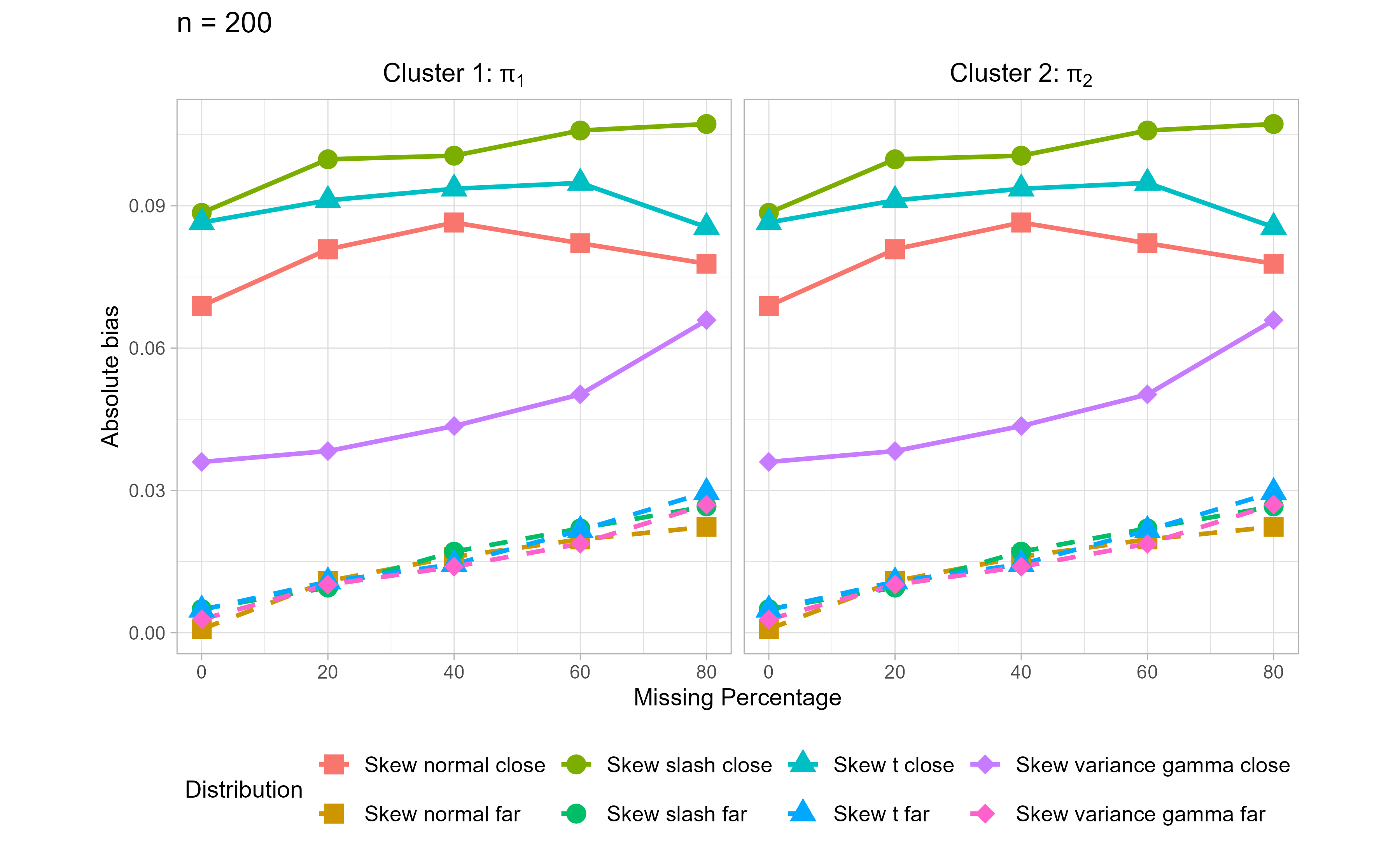}
    \caption{Average absolute bias for mixing proportions  across 200 replications for datasets of size $n = 200$, randomly generated from the distributions specified in the headers of each subplot.}
         
    \end{subfigure}
    \begin{subfigure}{0.49\textwidth}
         \includegraphics[width=\textwidth]{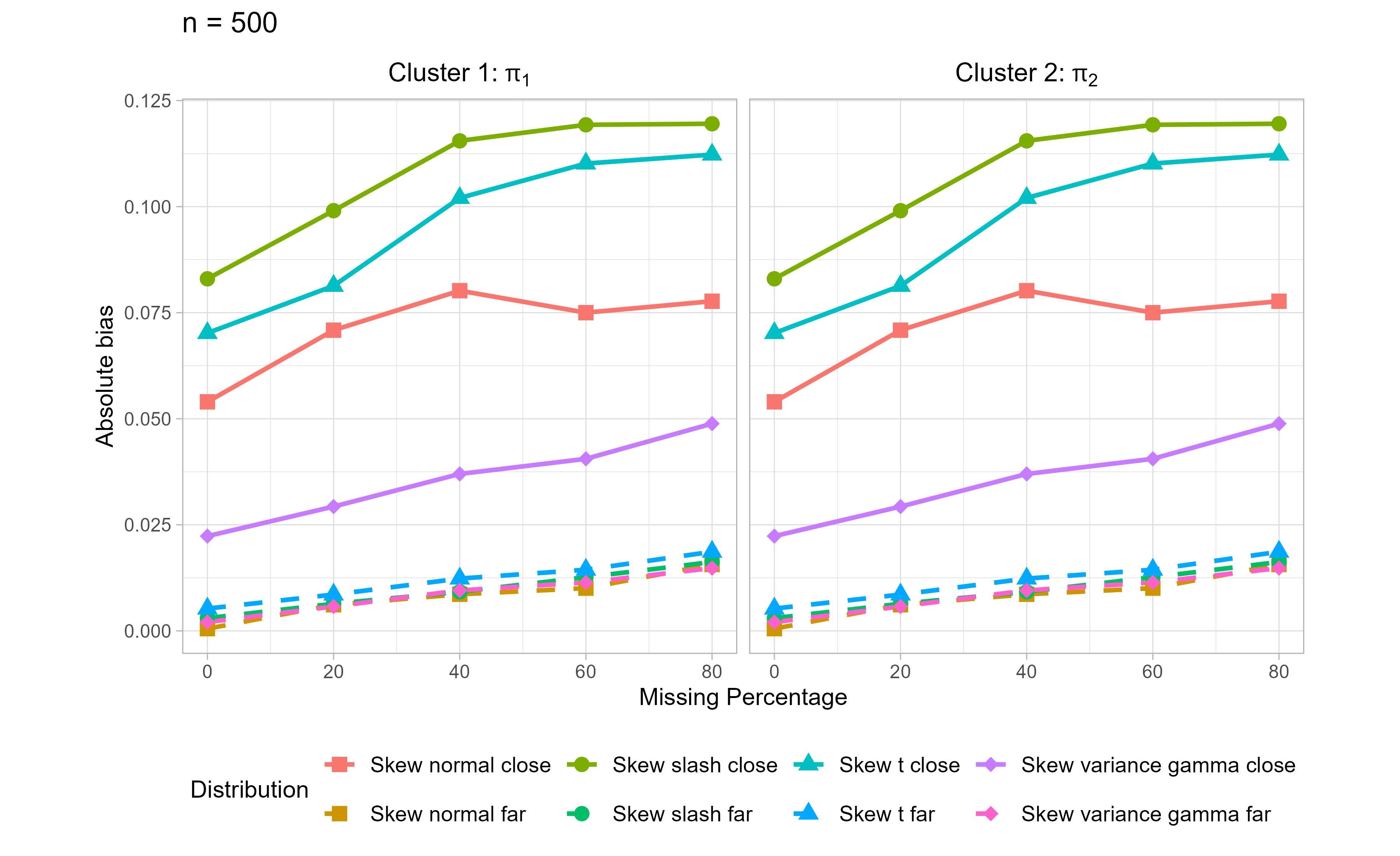}
         \caption{Average absolute bias for mixing proportions  across 200 replications for datasets of size $n = 500$, randomly generated from the distributions specified in the headers of each subplot.}
    \end{subfigure}
    \caption{ Average absolute bias for mixing proportions across 200 replications}
    \label{pi_abs_bias_200}
\end{figure}

\begin{figure}[H]
    \centering
    \begin{subfigure}{0.49\textwidth}
        \includegraphics[width=\textwidth]{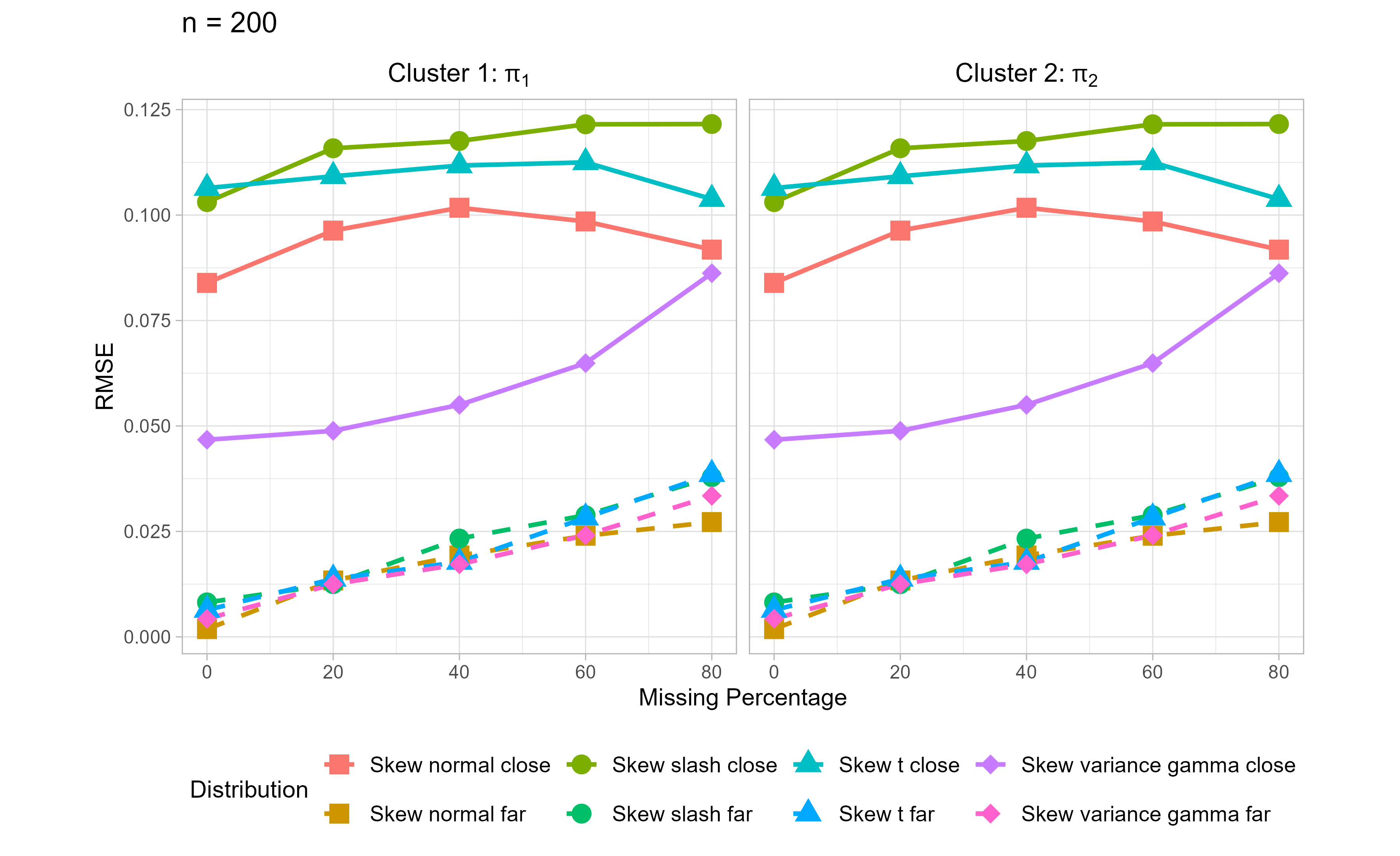}
         \caption{Average RMSE for mixing proportions  across 200 replications for datasets of size $n = 200$, randomly generated from the distributions specified in the headers of each subplot.}
    \end{subfigure}
        \begin{subfigure}{0.49\textwidth}
        \includegraphics[width=\textwidth]{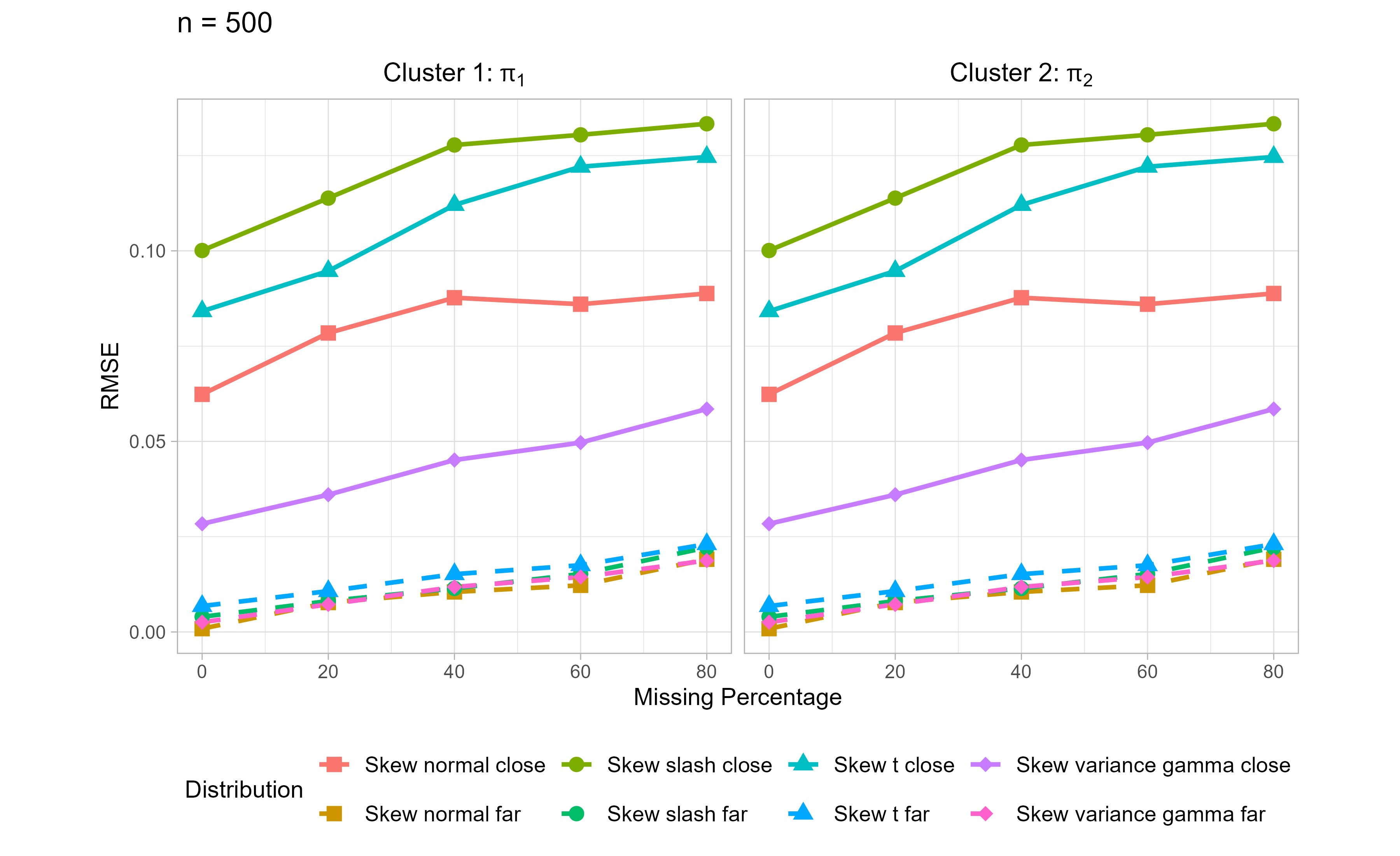}
         \caption{Average RMSE for mixing proportions  across 200 replications for datasets of size $n = 500$, randomly generated from the distributions specified in the headers of each subplot.}
    \end{subfigure}
            \caption{ Average RMSE for mixing proportions across 200 replications.}
    \label{pi_rmse_200}
\end{figure}

\begin{figure}[H]
    \centering
    \begin{subfigure}{0.49\textwidth}
        \includegraphics[width=\textwidth]{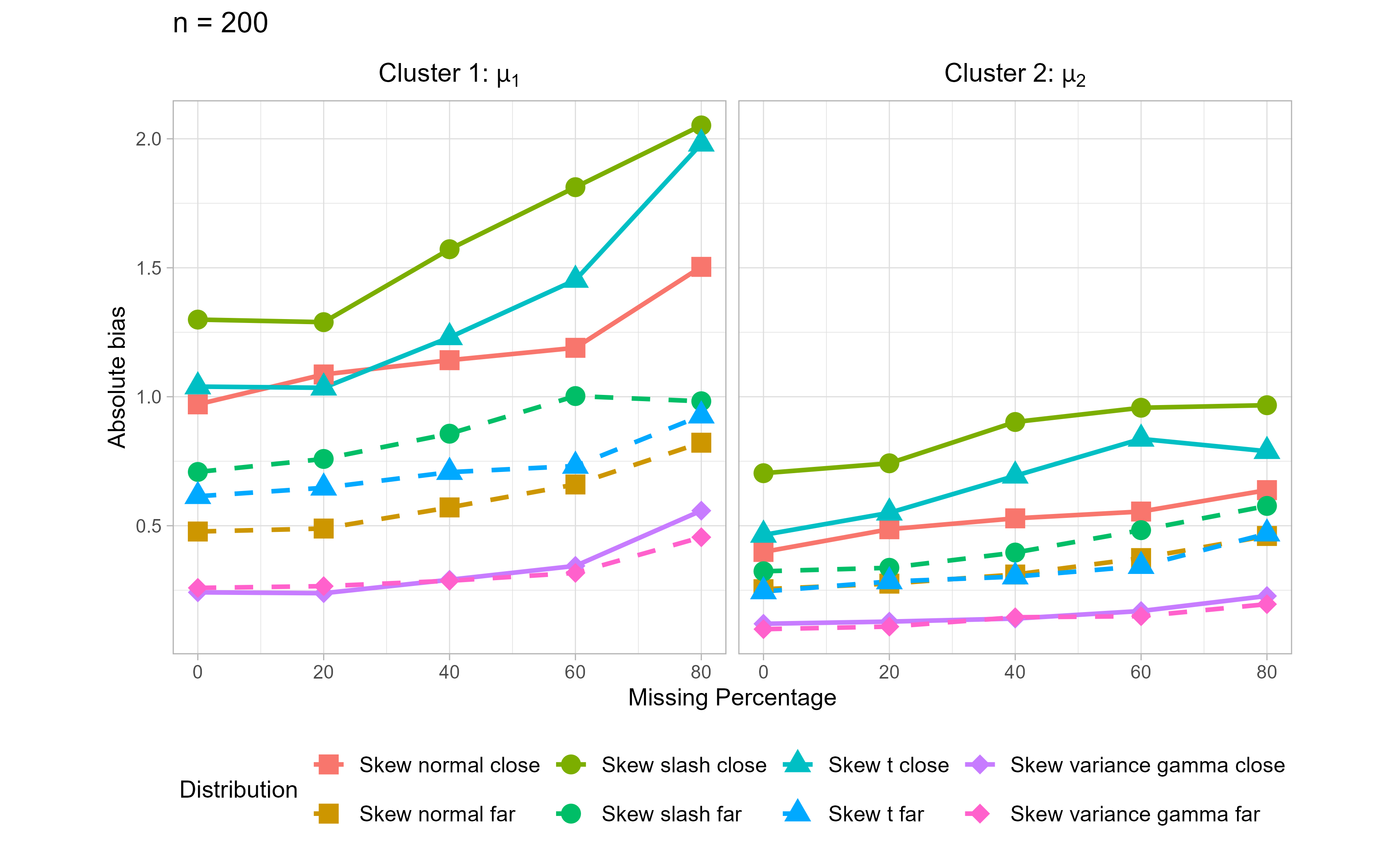}
         \caption{ Average absolute bias for location vectors  across 200 replications for datasets of size $n = 200$, randomly generated from the distributions specified in the headers of each subplot.}
    \end{subfigure}
    \begin{subfigure}{0.49\textwidth}
        \includegraphics[width=\textwidth]{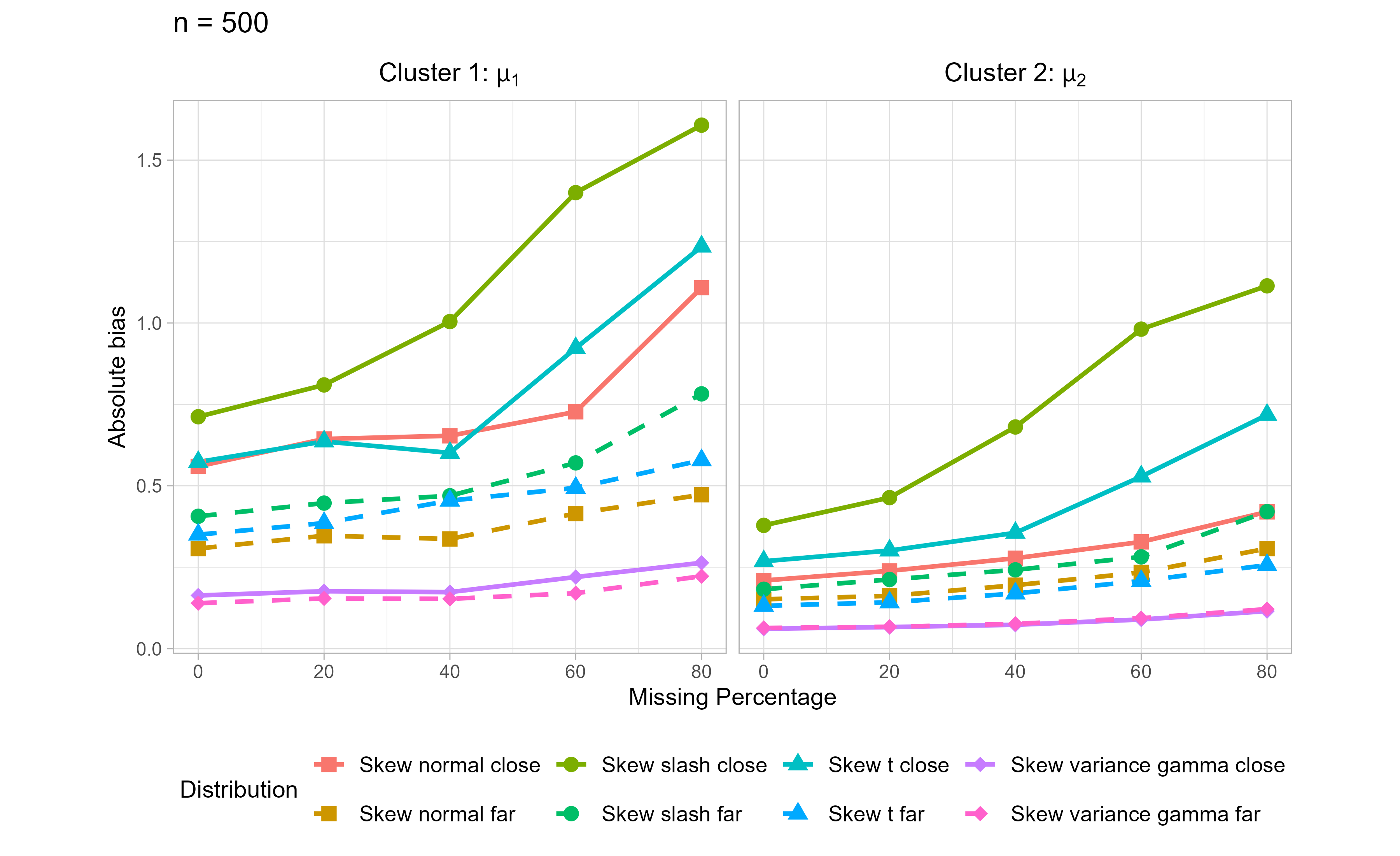}
         \caption{Average absolute bias for location vectors across 200 replications for datasets of size $n = 500$, randomly generated from the distributions specified in the headers of each subplot.}
    \end{subfigure}
    \caption{ Average absolute bias for location vectors across 200 replications.}
         \label{mu_abs_bias_200}
\end{figure}

\begin{figure}[H]
    \centering
    \begin{subfigure}{0.49\textwidth}
        \includegraphics[width=\textwidth]{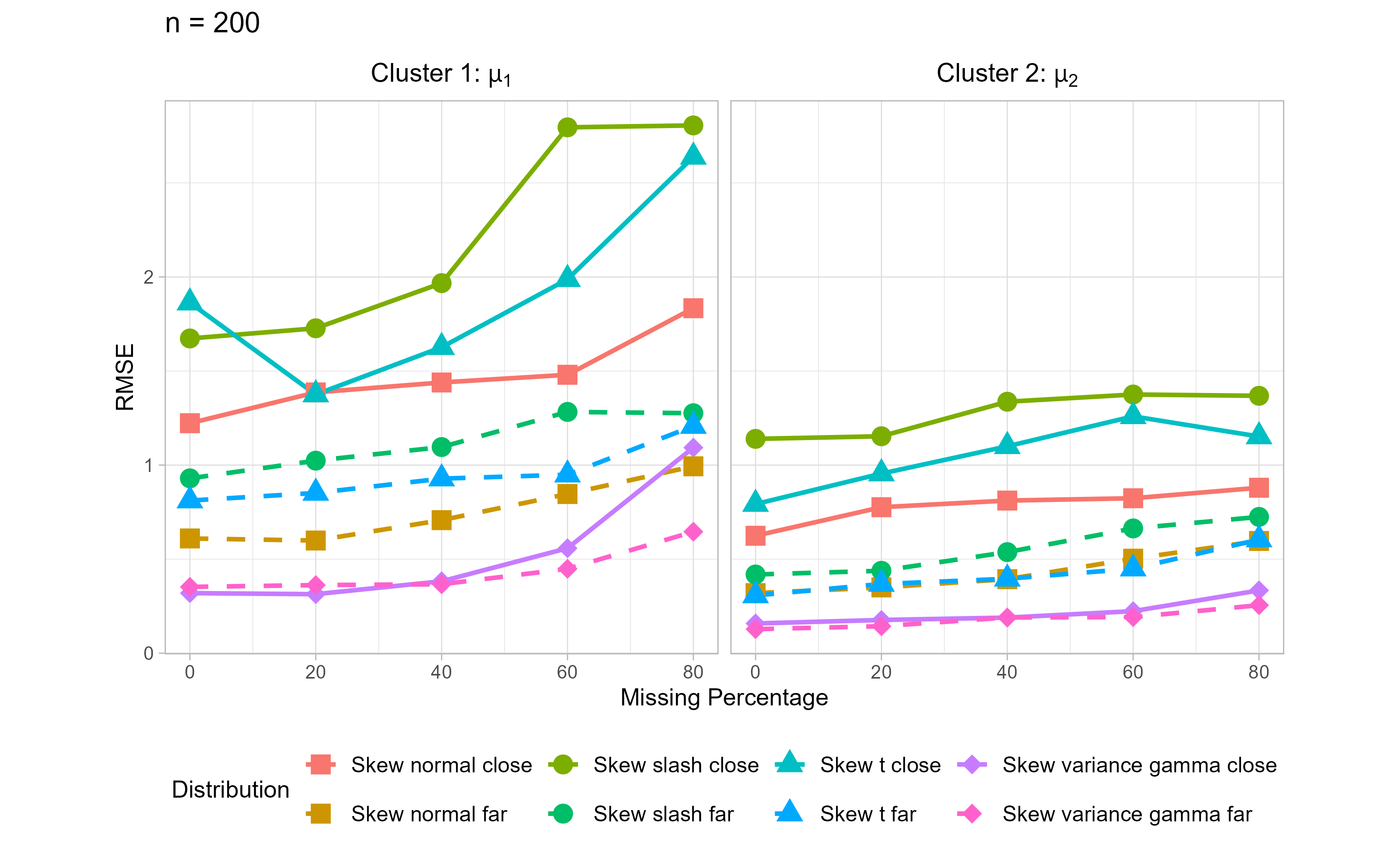}
         \caption{Average RMSE for location vectors  across 200 replications for datasets of size $n = 200$, randomly generated from the distributions specified in the headers of each subplot.}
    \end{subfigure}
        \begin{subfigure}{0.49\textwidth}
        \includegraphics[width=\textwidth]{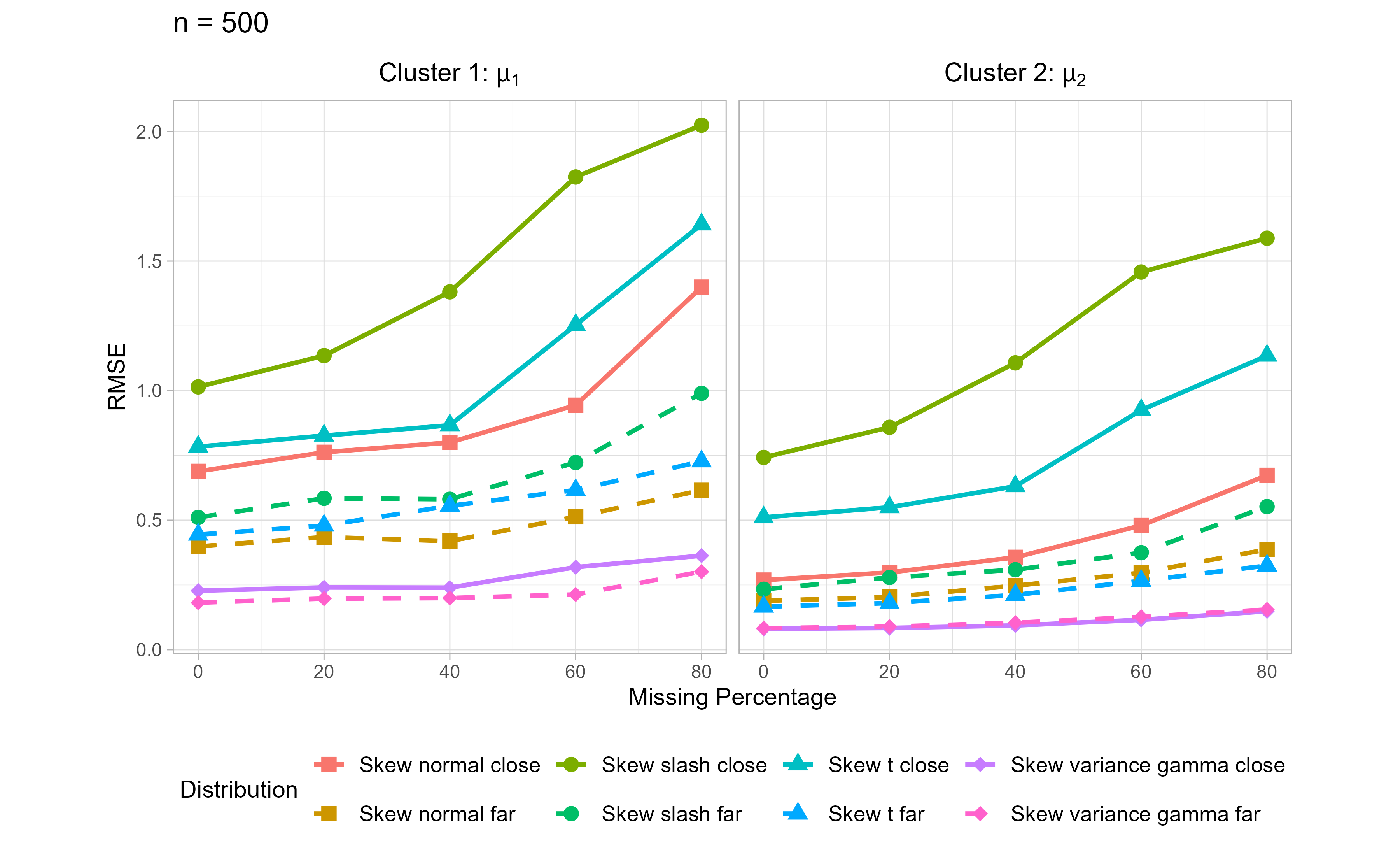}
         \caption{Average RMSE for location vectors  across 200 replications for datasets of size $n = 500$, randomly generated from the distributions specified in the headers of each subplot.}
    \end{subfigure}
        \caption{ Average RMSE for location vectors across 200 replications.}
         \label{mu_rmse_200}
\end{figure}


\begin{figure}[H]
    \centering
    \begin{subfigure}{0.49\textwidth}
        \includegraphics[width=\textwidth]{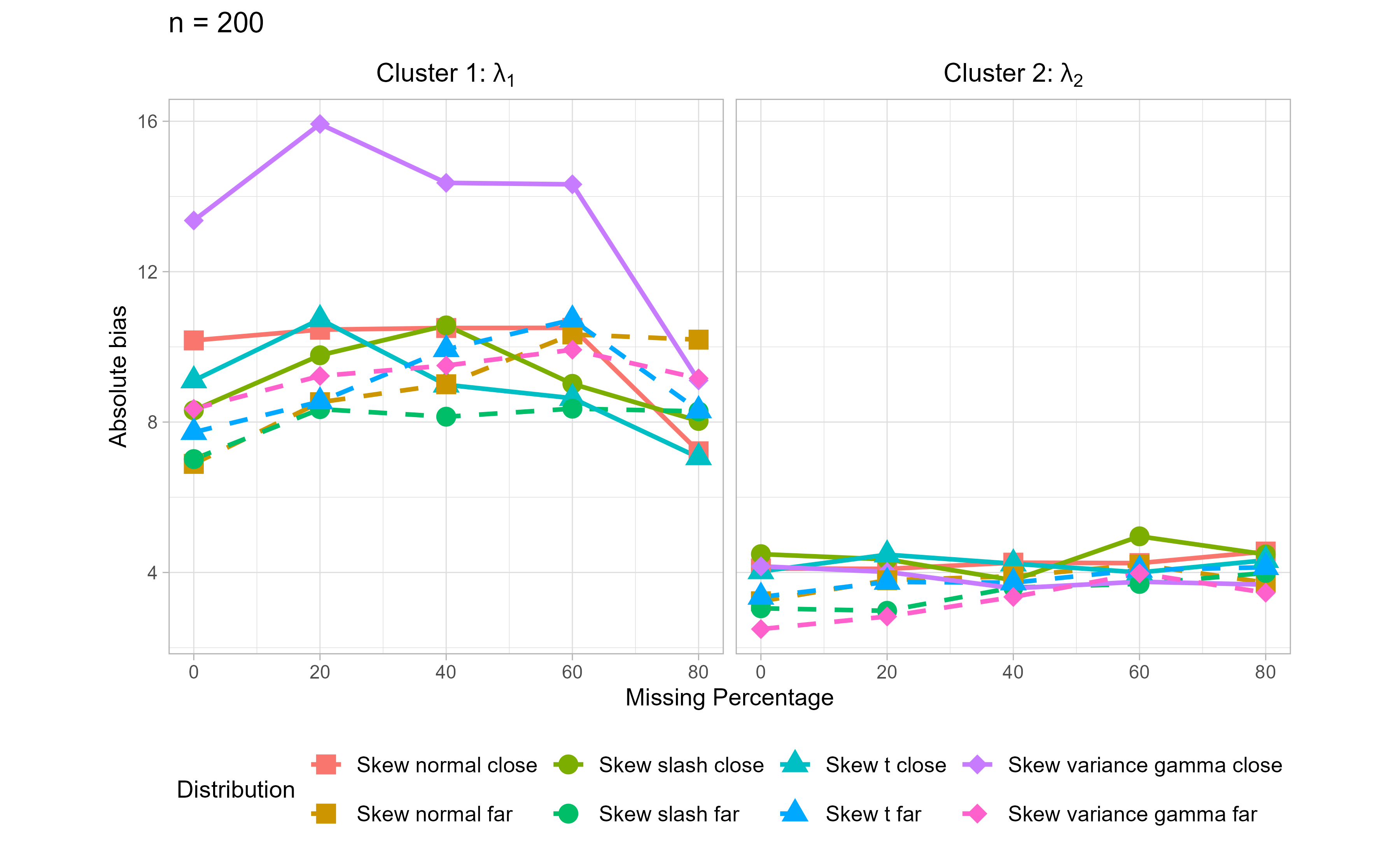}
         \caption{ Average Absolute bias for skewness vectors across 200 replications for datasets of size $n = 200$, randomly generated from the distributions specified in the headers of each subplot.}
    \end{subfigure}
    \begin{subfigure}{0.49\textwidth}
         \includegraphics[width=\textwidth]{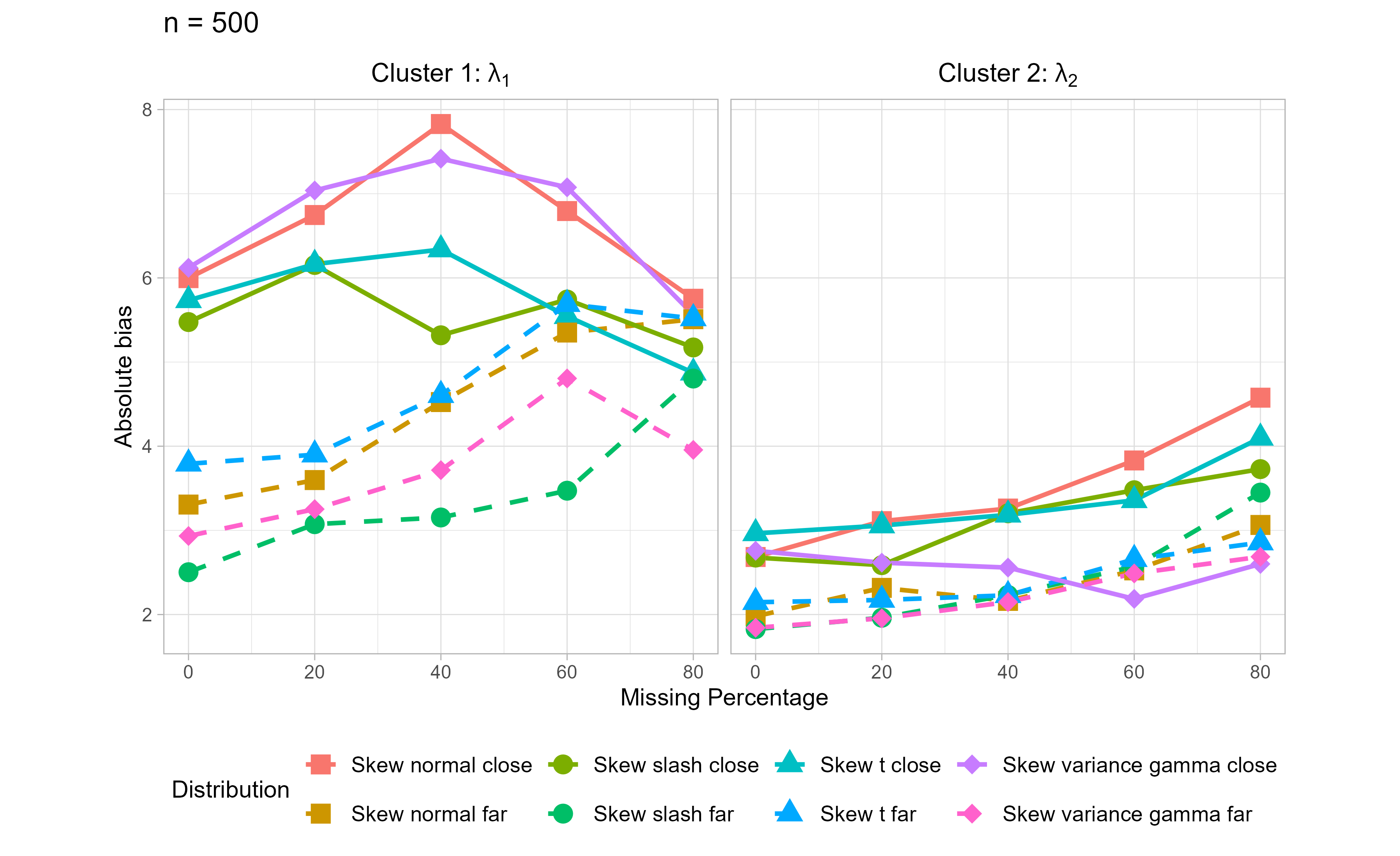}
         \caption{Average Absolute bias for skewness vectors across 200 replications for datasets of size $n = 500$, randomly generated from the distributions specified in the headers of each subplot.}
    \end{subfigure}
        \caption{ Average absolute bias for skewness vectors across 200 replications.}
         \label{lam_abs_bias_200}
\end{figure}

\begin{figure}[H]
    \centering
    \begin{subfigure}{0.49\textwidth}
        \includegraphics[width=\textwidth]{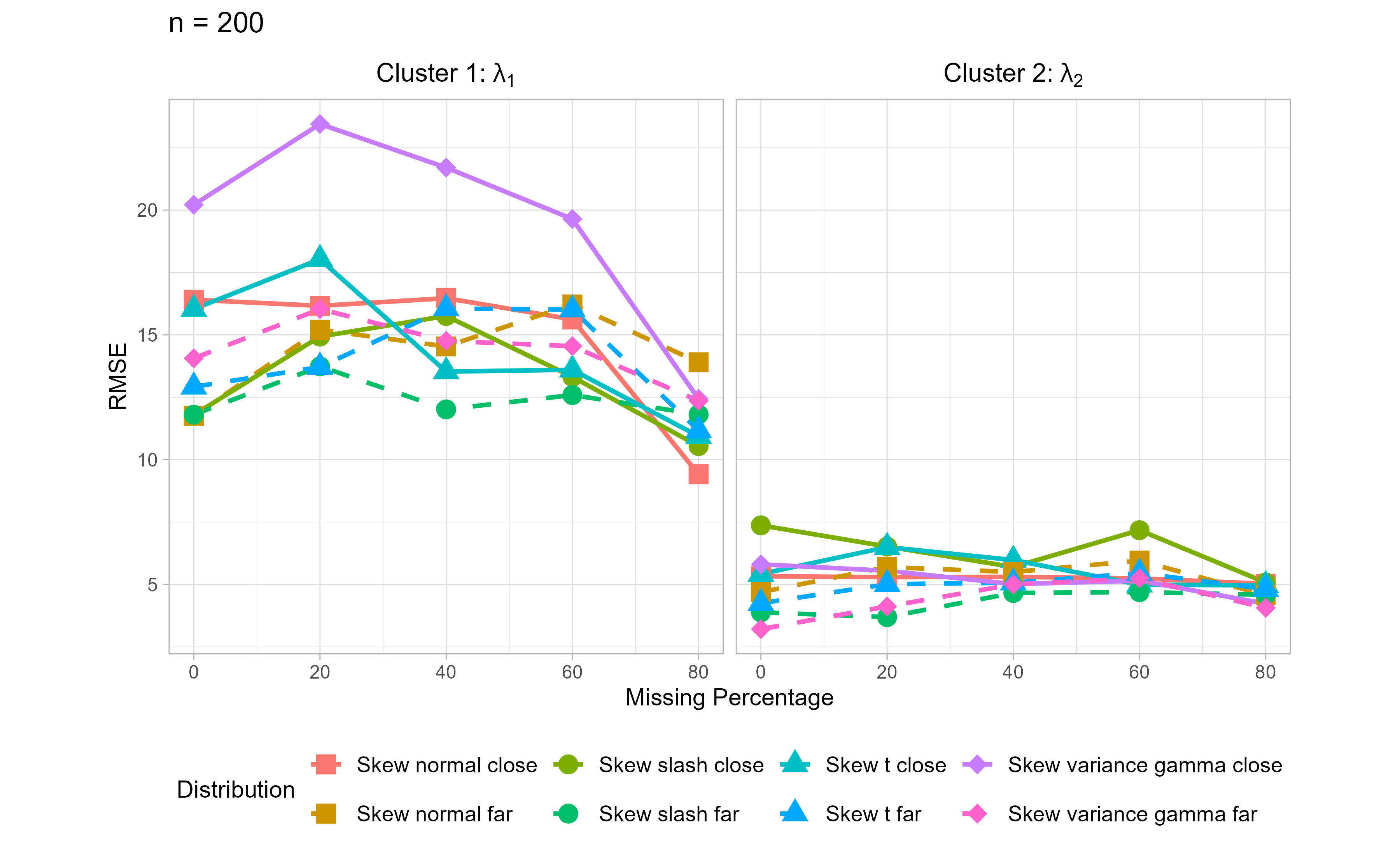}
         \caption{Average RMSE for skewness vectors across 200 replications for datasets of size $n = 200$, randomly generated from the distributions specified in the headers of each subplot.}
    \end{subfigure}
    \begin{subfigure}{0.49\textwidth}
        \includegraphics[width=\textwidth]{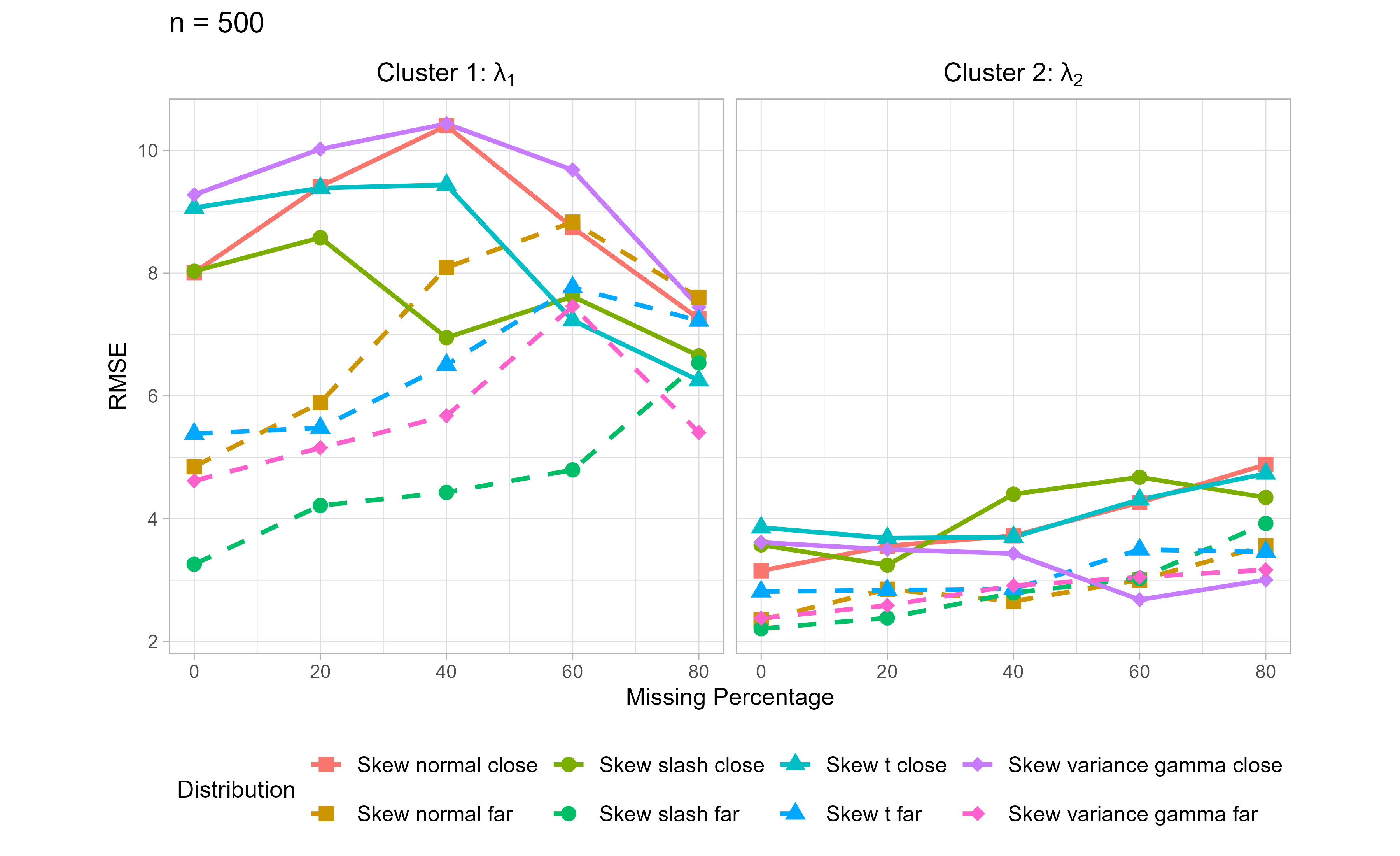}
         \caption{Average RMSE for skewness vectors across 200 replications for datasets of size $n = 500$, randomly generated from the distributions specified in the headers of each subplot.}
    \end{subfigure}     
        \caption{Average RMSE for skewness vectors across 200 replications.}
         \label{lam_rmse_200}
\end{figure}

\begin{figure}[H]
    \centering
    \begin{subfigure}{0.49\textwidth}
        \includegraphics[width=\textwidth]{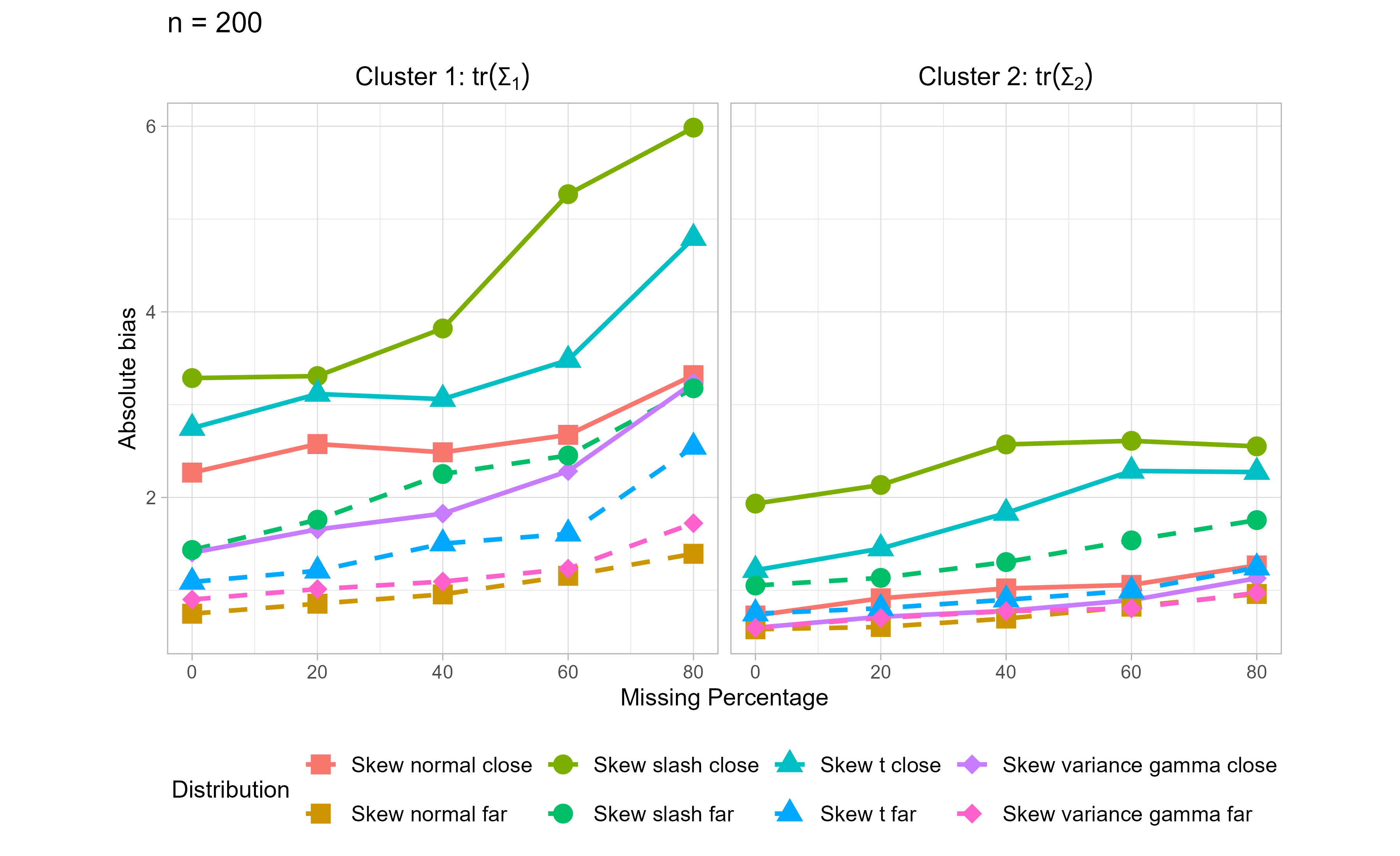}
         \caption{
         Average Absolute bias for scale trace across 200 replications for datasets of size $n = 200$, randomly generated from the distributions specified in the headers of each subplot.}
    \end{subfigure}
    \begin{subfigure}{0.49\textwidth}
        \includegraphics[width=\textwidth]{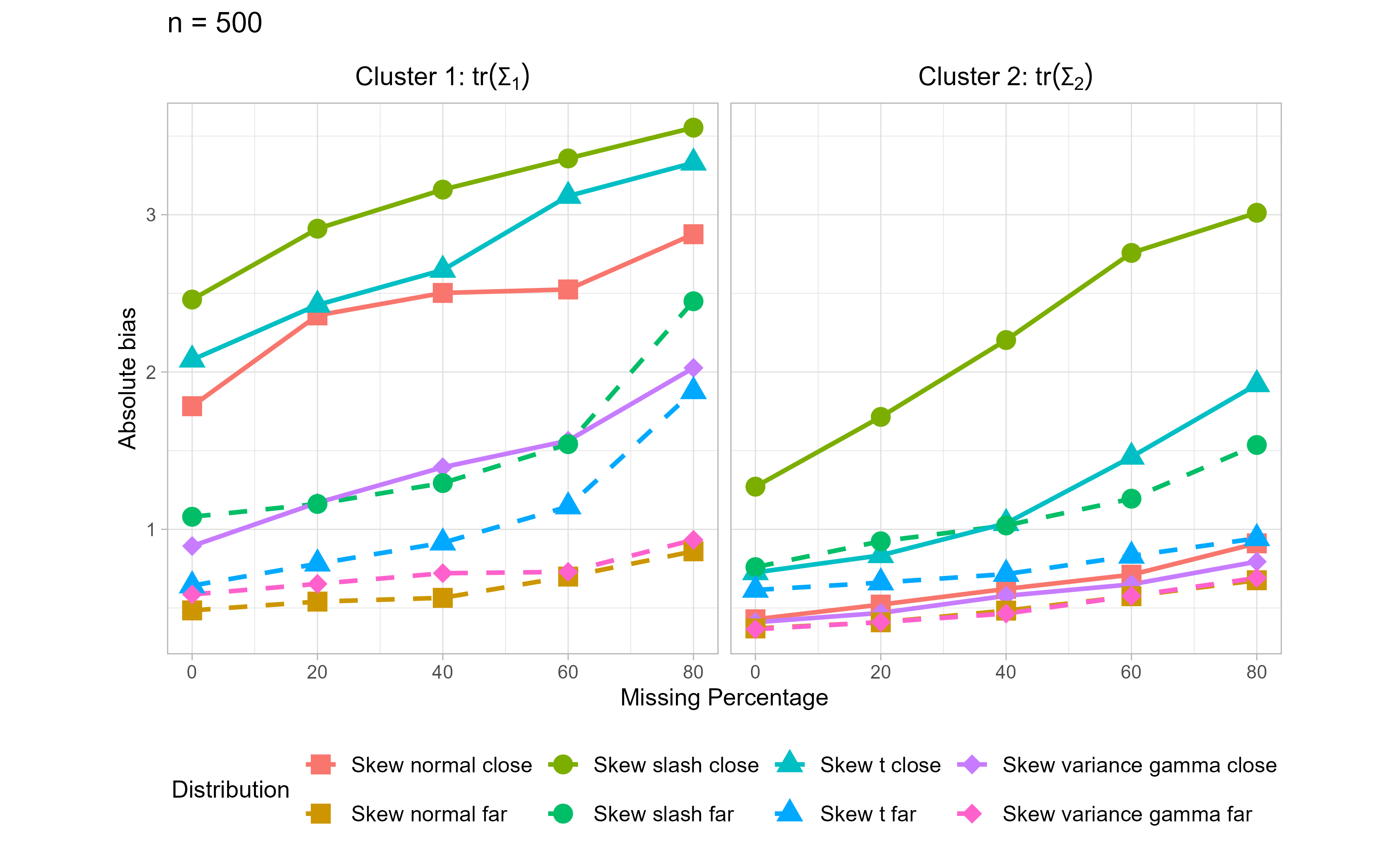}
         \caption{  Average Absolute bias for scale trace across 200 replications for datasets of size $n = 500$, randomly generated from the distributions specified in the headers of each subplot.}
    \end{subfigure}  
    \caption{ Average absolute bias for scale trace across 200 replications.}
         \label{var_abs_bias_200}
\end{figure}

\begin{figure}[H]
    \centering
    \begin{subfigure}{0.49\textwidth}
        \includegraphics[width=\textwidth]{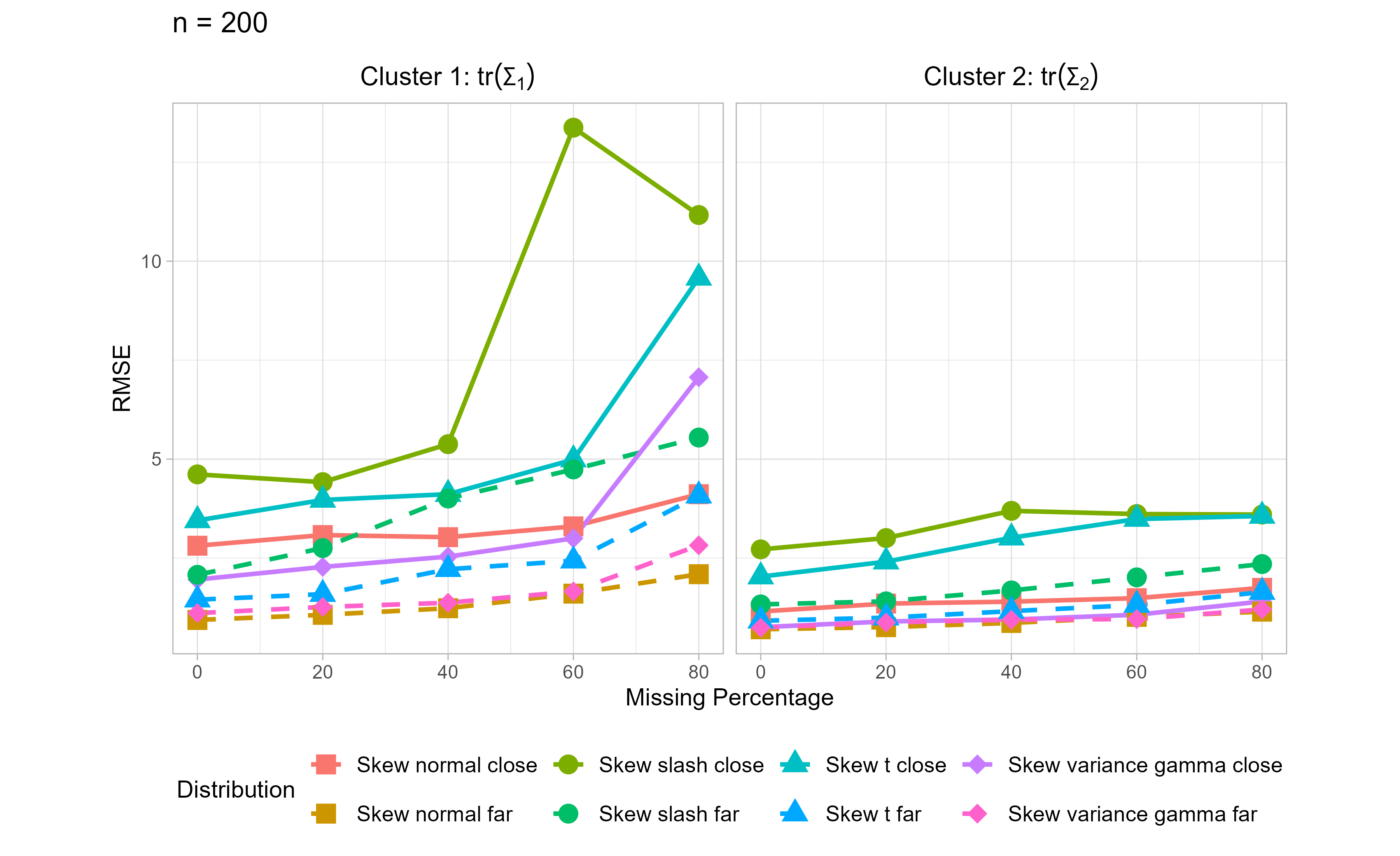}
         \caption{Average RMSE for scale trace across 200 replications for datasets of size $n = 200$, randomly generated from the distributions specified in the headers of each subplot.}
    \end{subfigure}
    \begin{subfigure}{0.49\textwidth}
        \includegraphics[width=\textwidth]{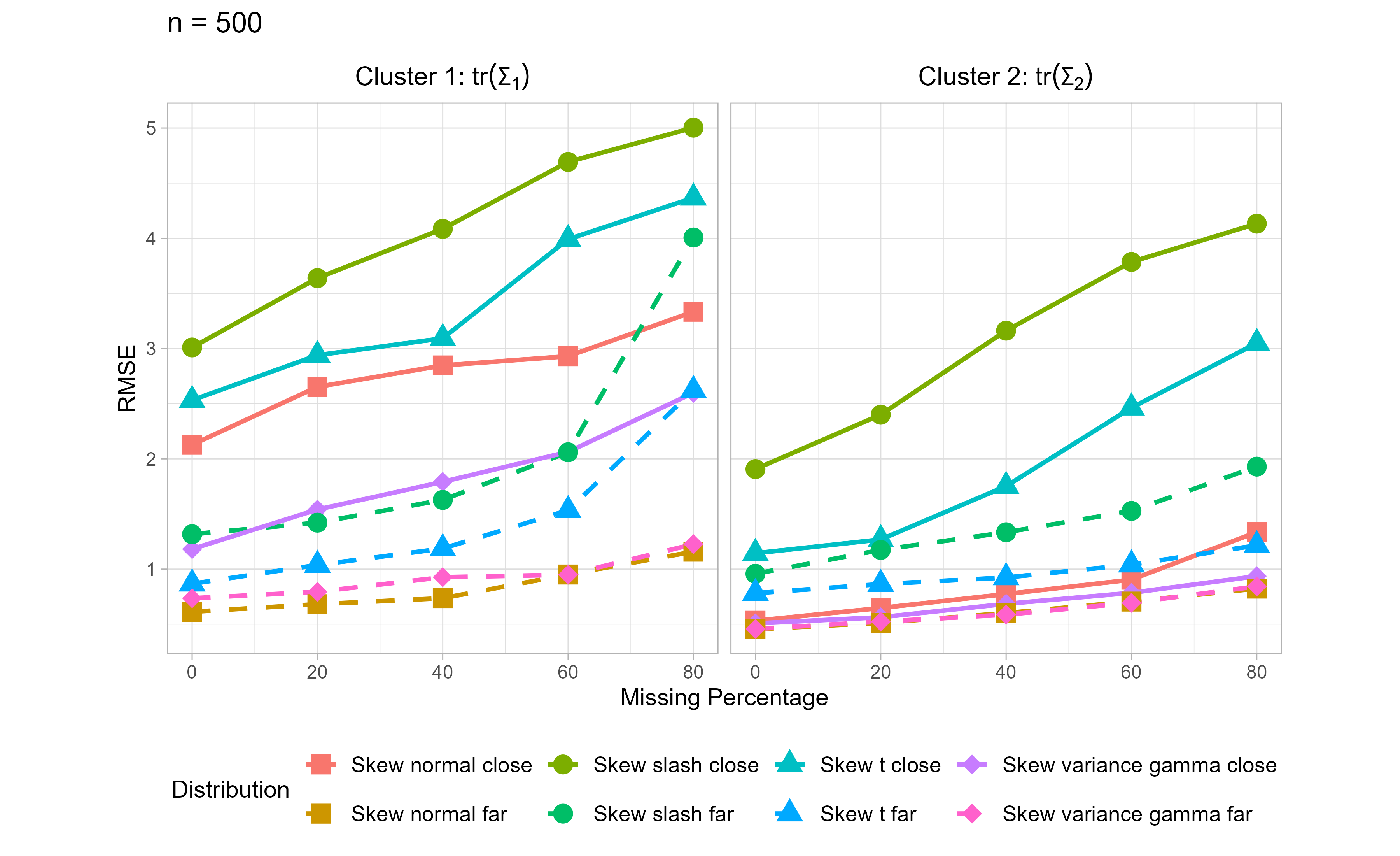}
         \caption{Average RMSE for scale trace across 200 replications for datasets of size $n = 500$, randomly generated from the distributions specified in the headers of each subplot.}
    \end{subfigure}
        \caption{ Average RMSE for scale trace across 200 replications.}
         \label{var_rmse_200}
\end{figure}


\begin{figure}[H]
    \centering
    \begin{subfigure}{0.49\textwidth}
        \includegraphics[width=\textwidth]{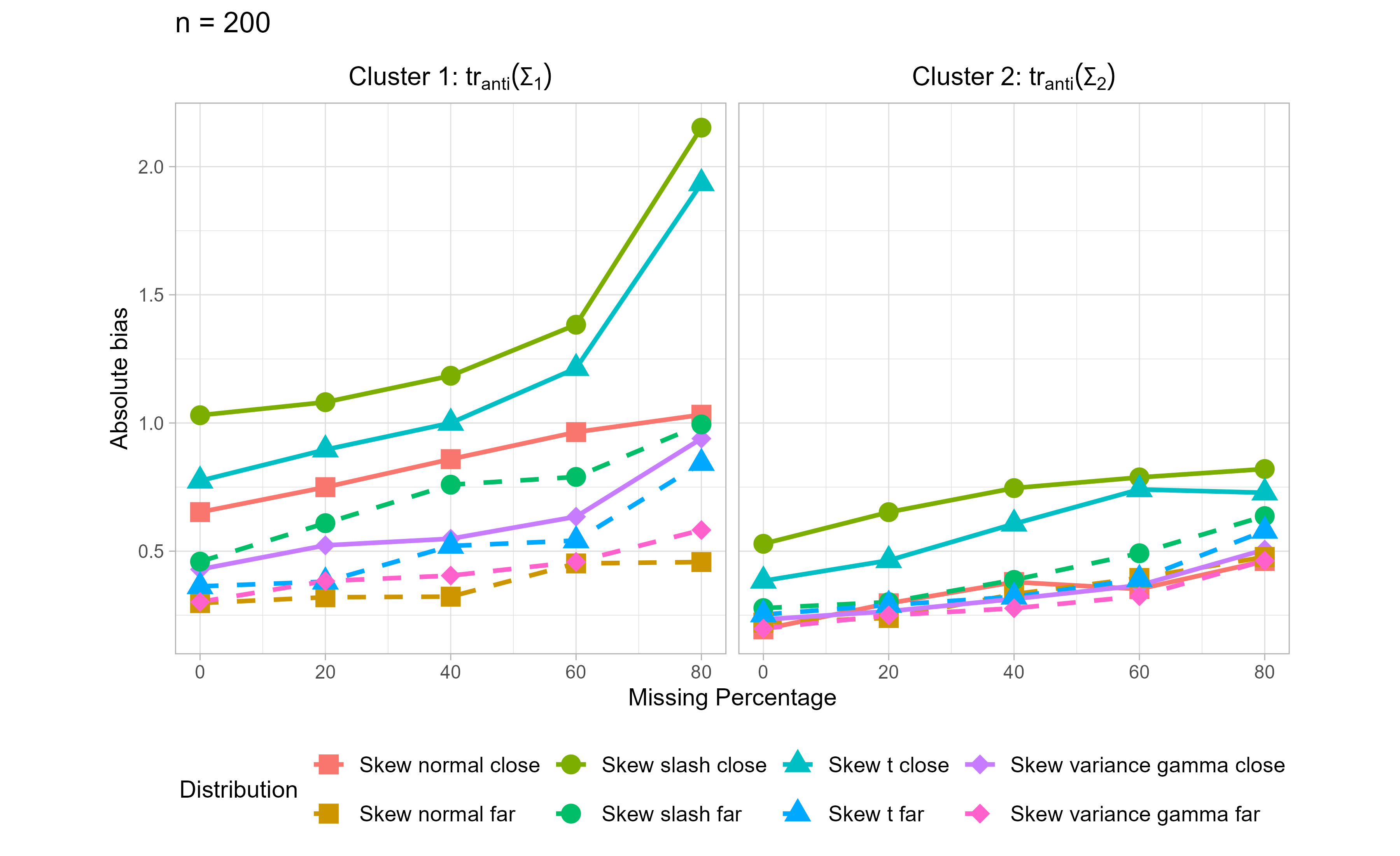}
         \caption{Average absolute bias for scale anti-trace across 200 replications for datasets of size $n = 200$, randomly generated from the distributions specified in the headers of each subplot.}
    \end{subfigure}%
    \begin{subfigure}{0.5\textwidth}
        \includegraphics[width=\textwidth]{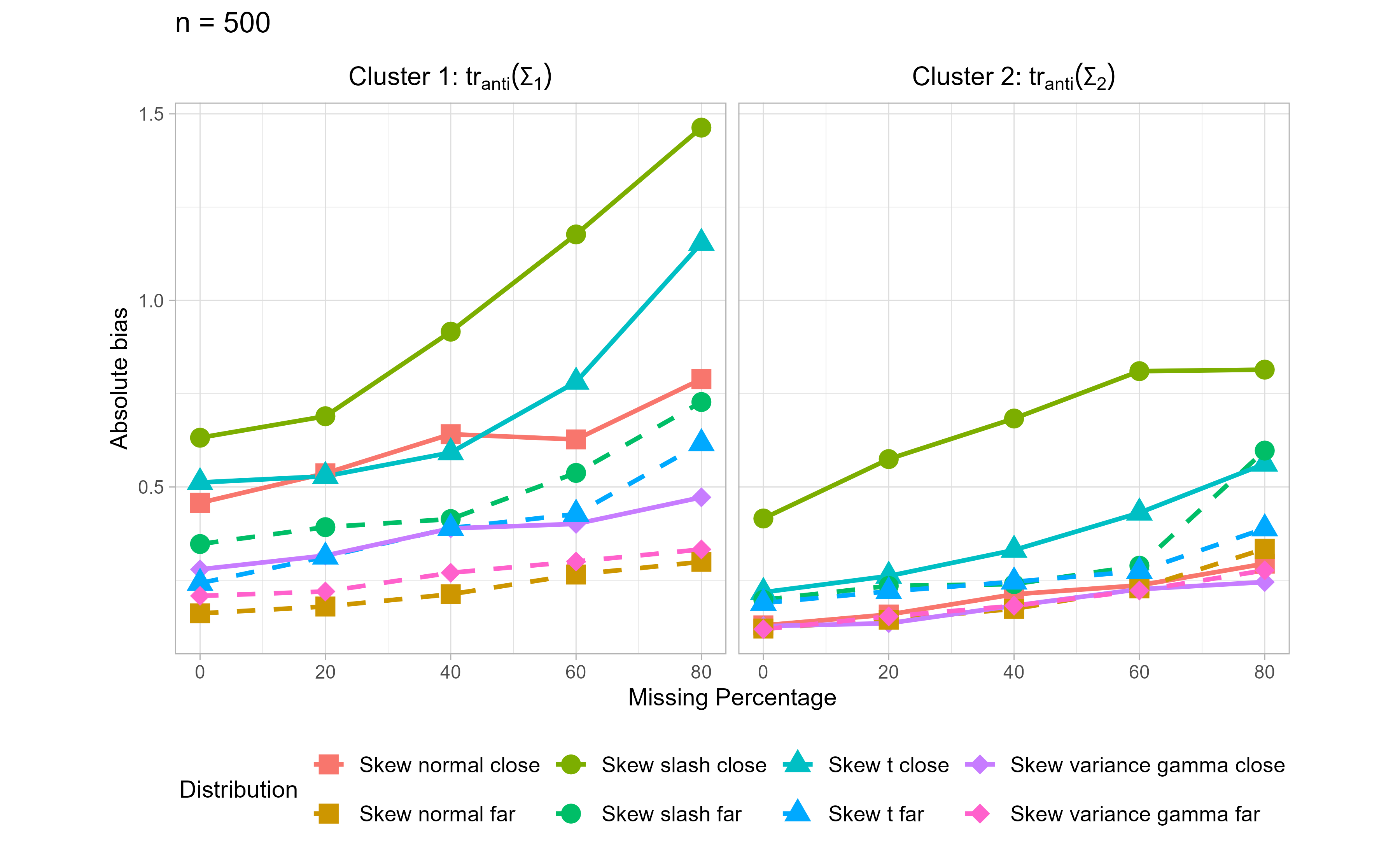}
         \caption{Average absolute bias for scale anti-trace across 200 replications for datasets of size $n = 500$, randomly generated from the distributions specified in the headers of each subplot.}
    \end{subfigure}
        \caption{ Average absolute bias for scale anti-trace across 200 replications.}
         \label{cov_abs_bias_200}
\end{figure}

\begin{figure}[H]
    \centering
    \begin{subfigure}{0.49\textwidth}
        \includegraphics[width=\textwidth]{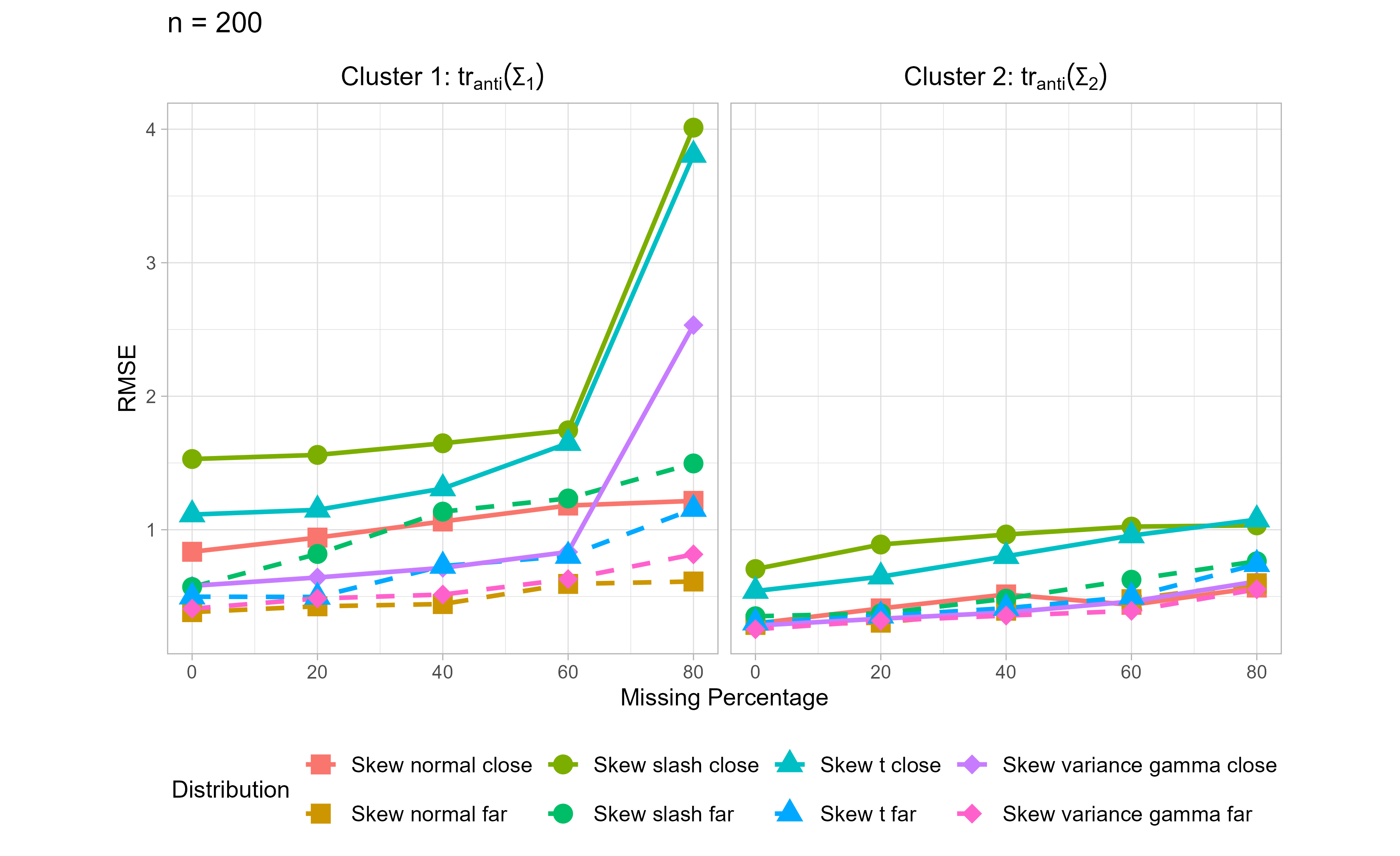}
         \caption{Average RMSE for scale anti-trace across 200 replications for datasets of size $n = 200$, randomly generated from the distributions specified in the headers of each subplot.}
    \end{subfigure}
        \begin{subfigure}{0.49\textwidth}
        \includegraphics[width=\textwidth]{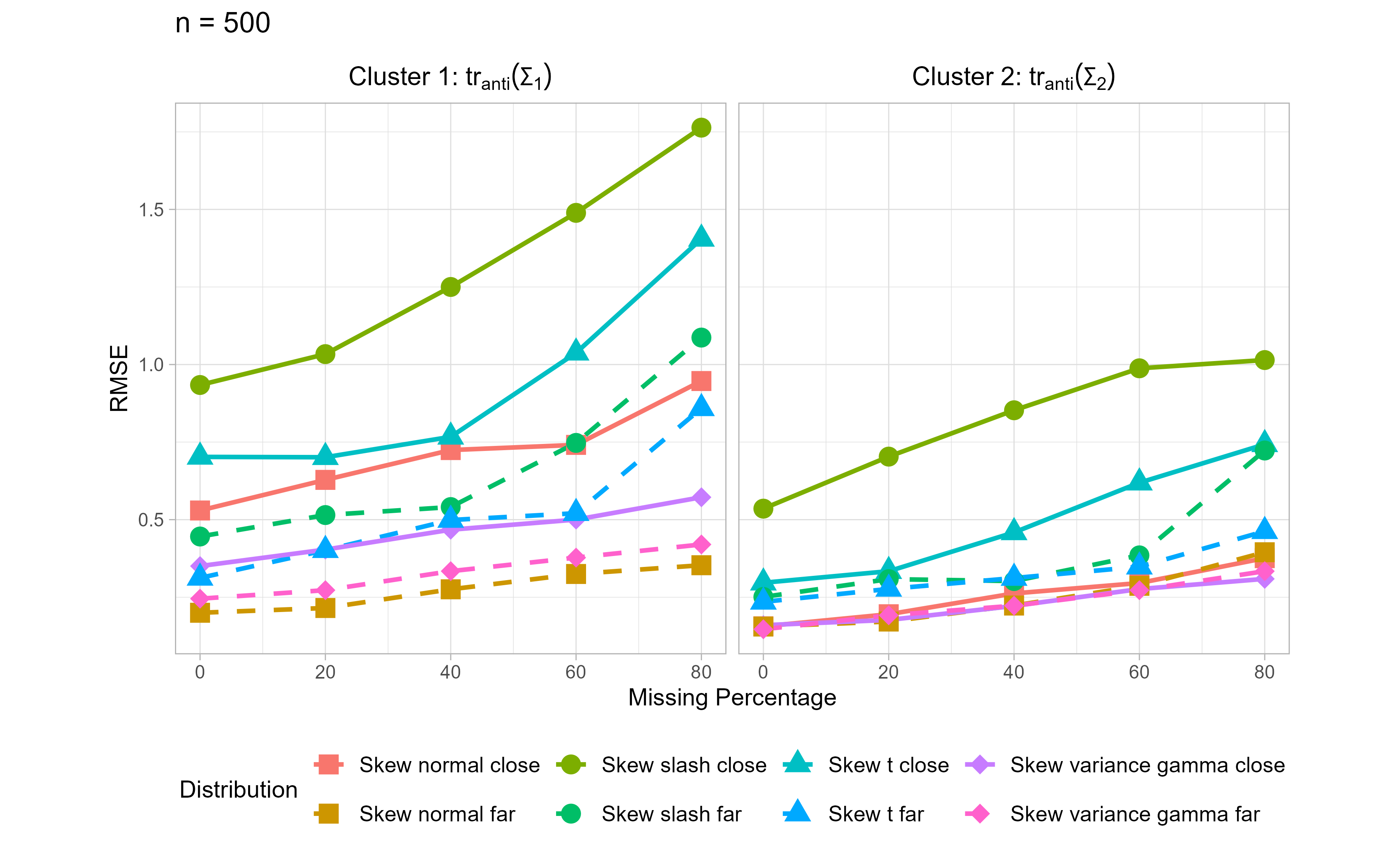}
         \caption{Average RMSE for scale anti-trace across 200 replications for datasets of size $n = 500$, randomly generated from the distributions specified in the headers of each subplot.}
    \end{subfigure}
            \caption{ Average RMSE for scale anti-trace across 200 replications.}
    \label{cov_rmse_200}
\end{figure}

A consistent trend observed is that parameter recovery deteriorates as the percentage of missing values increases, as expected. This effect is further influenced by the proximity between clusters, which is worsened when clusters are close and improves when they are well separated, supporting the expectation that a greater amount of information enhances estimation performance. Additionally, a performance hierarchy is evident among the fitted distributions. 
The multivariate skew-variance-gamma distribution consistently demonstrates the lowest total RMSE and total absolute bias across varying percentages of missing data, dataset sizes, and cluster proximities. 
It is followed by the multivariate skew-normal, multivariate skew-t, and finally the multivariate skew-slash distribution, which shows the poorest recovery. This hierarchy becomes more pronounced as cluster proximity decreases. Solid lines in the plots, representing the performance on closer clusters, highlight more prominent performance gaps compared to the relatively smaller differences represented by the dashed lines, where the dashed lines refer to well-separated clusters. Skewness, $\blam$, is the only parameter that deviates from this overall trend, as can be seen from Figures \ref{lam_abs_bias_200} and \ref{lam_rmse_200} . This counterintuitive behaviour arises specifically under close cluster proximity and composition. 

For both $n = 200$ and $n = 500$, closer clusters display an arched pattern in total absolute bias and a generally decreasing total RMSE for cluster 1. 
This may be attributed to the inherently smaller proportion of skewed observations in the dataset, and the added disadvantage posed by a smaller cluster proportion (e.g., $\pi_1 = 0.3$), which further hinders parameter recovery. 
As missingness increases, the reliance on expected values in the ECM algorithm grows, potentially improving skewness estimation under these constraints. However, when clusters are well separated, $\blam$ recovery aligns with the trends of the other parameters regardless of sample size. Interestingly, $\blam$ exhibits a reversed performance hierarchy compared to other parameters: the skew-slash distribution yields the lowest total absolute bias and RMSE, while the skew-variance-gamma performs worst in this specific case. The increased complexity of skew-variance-gamma may explain why $\blam$ is neglected, compared to multivariate skew-t and skew-slash.

While performance in the presence of missing data is not as optimal as in a complete dataset, it still gives meaningful and informative results that would otherwise be inaccessible. An alternative to fitting distributions using the ECM algorithm would involve restricting analysis to the complete subset of the data; however, this would reduce the information available to the ECM algorithm, introducing unnecessary bias and uncertainty into the results.

\section{Data application}
\label{application}

The study of CO2 emissions is a long and ongoing trend in environmental research. 
Numerous human activities such as energy consumption, deforestation, manufacturing, transportation, among others have increased the concentration of CO2 in the atmosphere thus enhancing possible effects on climate. 
Several research articles suggest a causal relationship between the economic growth of a country and CO2 emissions through its economic growth (\cite{climatecovid,co2export, co2Africa, co2cee}). Statistical analyses have highlighted the long run economic benefits from abandoning fossil fuels in developing nations. There is statistical support for research in renewable energy consumption (\cite{albijanic2023analysis, statsclimate, dynamictimeseries, co2ann}) that can be useful in policy shaping towards a more sustainable future.

\subsection{Data description}

EDGARv8.1 database containing an assortment of climate data provides emissions datasets for each of the following air pollutants: SO2, NOx, CO2, NMVOC, NH3, PM10, PM2.5, BC, and OC (\cite{EDGAR}). 
Records provided by EDGAR start from 1970 and currently end at 2022 and are categorised by source, sector, and country. 
Time steps are provided on both a monthly and an annual basis. Source here refers to whether the pollutant originates from fossil fuels or from bio mass. Note that EDGAR excludes emissions from large scale biomass burning with Savannah burning, forest fires, and sources and sinks from land-use, land-use change and forestry. We focus on the fossil fuel CO emitted in 2022 across all available countries and focus on the following sectors: (1) Main activity electricity and heat production, (2) Manufacturing industries and construction, (3) Road transportation no resuspension, (4) Residential and other sectors, (5) Oil and natural gas, (6) Lime production, and the (7) Metal industry. 
Some countries do not have records available for all sectors. 
In particular, 84.13\% of the dataset's rows are incomplete. Column wise percentages of missing values are given in Table \ref{colwise missing}.
   \begin{table}[H]
   \centering
    \caption{Proportion of missing values by sector in the EDGARv8.1 dataset.}
   \begin{tabular}{clS}
   \toprule
    Variable & Name                                          & {Missing (\% )}\\
    \midrule
    X1      & Main activity electricity and heat production  & 0.48  \\
    X2      & Manufacturing industries and construction      & 0.48  \\
    X3      & Road transportation no resuspension            & 0.00  \\
    X4      & Residential and other sectors                  & 0.48  \\
    X5      & Oil and natural gas                            & 60.1  \\
    X6      & Lime production                                & 55.3  \\
    X7      & Metal industry                                 & 66.3  \\
    \bottomrule
       \end{tabular}
       \label{colwise missing}
   \end{table}
 Some summary statistics of the dataset, expressed as Millions of tonnes (Mt) are provided in \tablename~\ref{sum stats}.
\begin{table}[H]
 \centering
 \caption{Sector-wise summary statistics for the EDGARv8.1 dataset.}
    \begin{tabular}{cSSSSS}
    \toprule
     Variable &   {Mean}   & {Standard deviation} & {IQR} & {Skewness} & {Kurtosis} \\
     \midrule
    X1        & 75.9487  & 400.6935  & 26.6228  & 10.6411 & 128.4535\\
    X2        & 214.6219 & 1965.7815 & 16.0866  & 13.0019 & 177.5671\\
    X3        & 211.5783 & 867.9373  & 116.1740 & 10.7015 & 134.5023\\
    X4        & 78.1599  & 599.2809  & 8.6795   & 12.9579 & 178.5300\\
    X5        & 10.4714  & 25.5350   & 6.7581   & 3.8640  & 18.7721\\
    X6        & 0.0088   & 0.0624    & 0.0022   & 9.3799  & 89.6343\\
    X7        & 785.3800 & 4590.0476 & 174.1204 & 8.0314  & 66.2949\\
    \bottomrule
    \end{tabular}
    \label{sum stats}
\end{table}

   \subsection{Results}
   
   Due to the large scale of the observations, we apply a log transformation prior to fitting each of the distributions across various predetermined numbers of clusters. This transformation serves to stabilise the variance across both observations and variables. 
   The performance of each mixture model is measured through the usual Bayesian Information Criterion (BIC). 
   \begin{align*}
       BIC &= 2\ln(L) - P\ln(n),
   \end{align*}
   where $L$, $P$ and $n$ denote the likelihood of the fitted model at convergence of the ECM for FMSMSN with MAR algorithm, the number of free parameters, and sample size, respectively. Model selection is determined through maximising the BIC. After fitting all four distributions on the datasets for varying numbers of clusters their observed log-likelihoods and BICs are given in \tablename~\ref{bic}.
   \begin{table}[H]
   \centering
   \caption{Number of clusters ($G$), log-likelihood values, and BIC values after running the ECM for FMSMSN with MAR algorithm.}
      \label{bic}
   \begin{tabular}{lcSS}
   \toprule
    Mixture component        &     $G$        & {log-likelihood } &  { BIC }          \\
    \midrule
    Skew-normal         &       1        &  -2140.416       & -4505.009          \\
    Skew-t              &       1        &  -2101.591       & -4432.696          \\
    Skew-slash          &       1        &  -2101.246       & -4432.006          \\
    Skew-variance-gamma &       1        &  -2108.206       & -4445.925          \\
    Skew-normal         &       2        &  -1998.665       & -4451.020          \\
    Skew-slash          &       2        &  -1982.555       & -4429.476          \\
    Skew-t              &       2        &  -1979.074       & -4422.514          \\
    Skew-variance-gamma &       2        & -1959.650        & -4383.66           \\
    Skew-normal         &       3        & -1908.094        & -4499.392          \\
    Skew-slash          &       3        & -1889.803        & -4478.824          \\
    Skew-t              &       3        & -1888.842        & -4476.902          \\
    Skew-variance-gamma &       3        &-1885.642         & -4470.502          \\
    \bottomrule
       \end{tabular}
   \end{table}


According to the highest BIC, a two-component skew-variance-gamma distribution performs best. 
The segmented transformed data according to the highest BIC value is visualised in \figurename~\ref{pairplots}. 
The estimated locations and skewness are given in \tablename~\ref{loc and skew}.
\begin{table}[H]
\centering
   \caption{Estimated location and skewness parameters for a 2-component skew-variance-gamma mixture.}
   \begin{tabular}{cS[round-precision = 3,table-number-alignment = left]S[ round-precision = 3,table-number-alignment = left]S[ round-precision = 3,table-number-alignment = left]S[ round-precision = 3,table-number-alignment = left]}
   \toprule
   Variable  &  {$\mu_1$} & {$\mu_2$}  & {$\lambda_1$}   & {$\lambda_2$}     \\
   \midrule
    X1        &  4.3169083 & -0.3223242 &  -2.5790454     &    7.1959850      \\
    X2        &  2.9419383 & -4.9014091 &  -0.5034683     &    13.2571581     \\
    X3        &  4.3446167 & 0.9744648  &  1.0385812      &    -0.6551722     \\
    X4        &  1.2437384 & -4.5575255 &  0.4218148      &    21.2454136     \\
    X5        &  0.7496323 & -1.0565692 &  -1.8228591     &    -8.1857921     \\
    X6        &  -7.198602 & -10.475639 & -1.2365834      &    2.2166617      \\
    X7        &  2.6864210 & -1.0090757 &  0.2186186      &    16.7984869     \\
    \bottomrule
    \end{tabular}
    \label{loc and skew}
\end{table}

\begin{figure}[H]
    \centering
    \includegraphics[width=0.9\linewidth, height=0.9\linewidth]{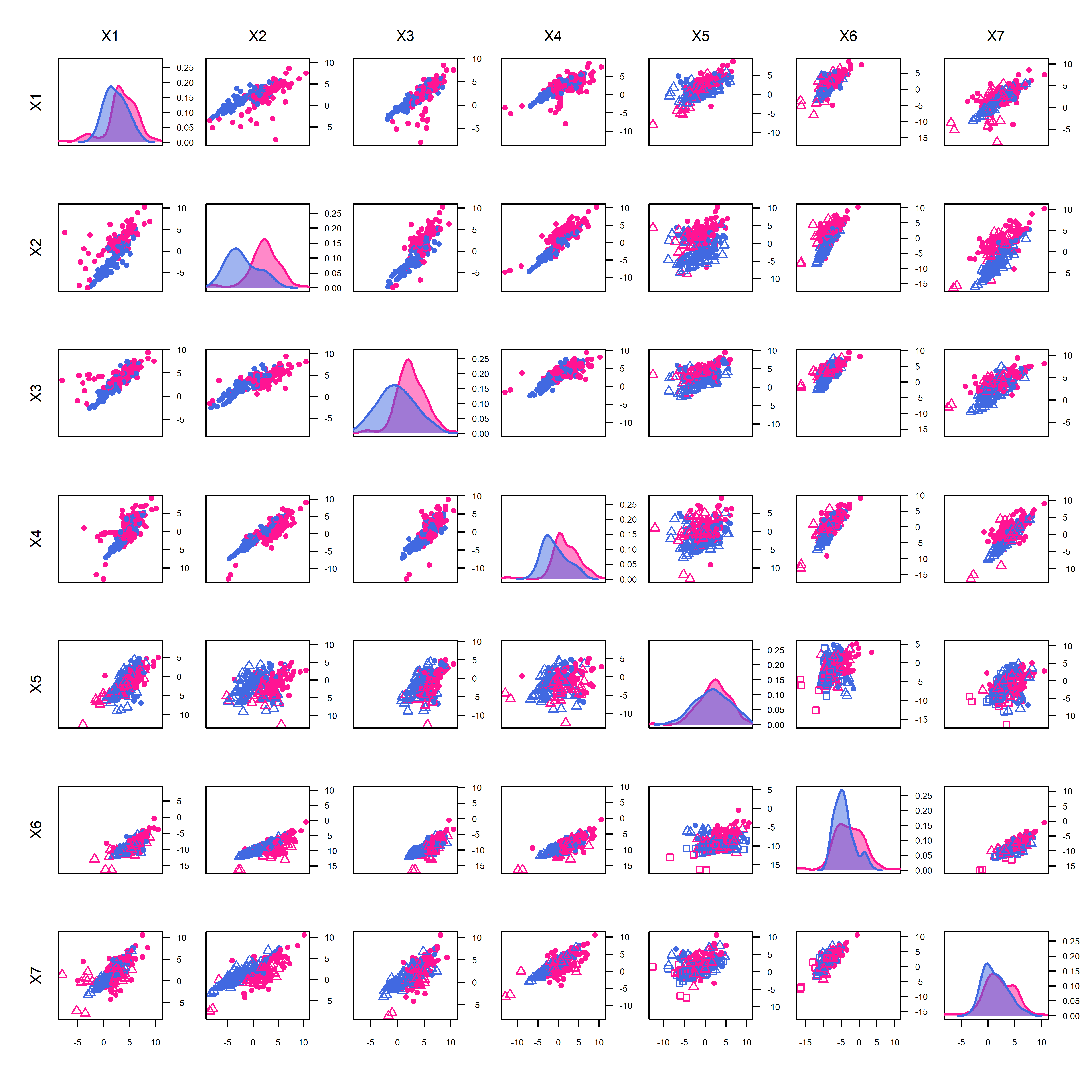}
    \caption{Pairwise scatter plot matrix for the EDGARv8.1 data, with colours indicating cluster assignments from a fitted mixture of two skew-variance-gamma distributions. Solid circles represent observations with no missing values in the variable pair. 
    Hollow triangles mark pairs where one variable was imputed, and hollow squares indicate that both values were imputed. Imputations reflect the final E-step of the EM algorithm.}
    \label{pairplots}
\end{figure}

\begin{figure}[H]
    \centering
    \includegraphics[width=\linewidth]{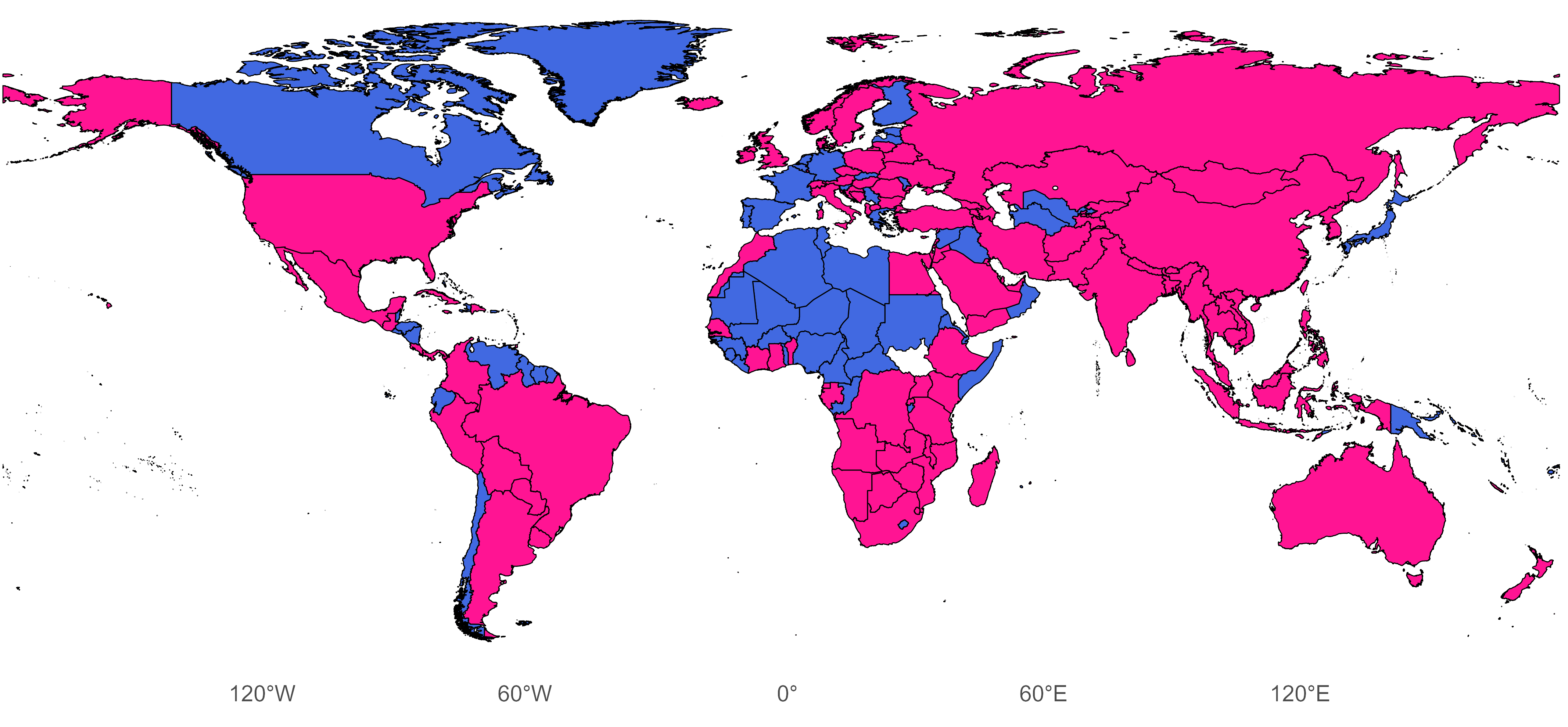}
    \caption{Countries coloured based on clusters from a mixture of two skew-variance-gamma distributions.}
    \label{fig:enter-label}
\end{figure}

The estimated locations $\bmu_1$ and $\bmu_2$ demonstrate that countries can be grouped into an overall lowered or higher volume of CO2 emissions. The second cluster's skewness vector has large magnitudes, implying that, while its countries generally emit fewer tonnes of CO2, there are a few extreme/atypical recorded values from the cluster's location. In contrast, the estimated skewness parameter for the first cluster is estimated closer to zero, implying that less skewness is present. The estimated mixing probabilities suggest that 46\% of the available countries have lower CO2 emissions across the seven sectors under consideration. The composition of this group includes most of North Africa, Lesotho, Canada, Japan, South Korea, Indonesia, Chile, Venezuela, Oman, Turkmenistan, Uzbekistan, Oman, Iraq, Syria, Greece, Finland and the stretch of countries with shared borders from Spain to Germany. These have so-called moderately developed to developed economies, boasting a corresponding GDP per capita and total GDP. Road transportation sector emitted the most CO2 emissions in 2022 for this cohort (\cite{climatecovid}).

In contrast, 54\% of the available countries emitted more CO2 across the seven sectors with road transportation and electricity and heat being the leading contributors. This cluster includes majority of the Global South, the United States, Australia, and the United Kingdom. This cluster suggests that countries with developing economies and improving GDPs come at the cost of increased CO2 emissions. Notice also that this cluster includes populous powered economies such as China, India, and the United States, suggesting population size is a key contributor. Southern Africa, and most of South America have great potential for economic growth. These countries’ development is influenced by specific geographical factors, but they also rely on rich human and natural resources and trade (\cite{gunaratne}).

Notice here that with more than 84\% of incomplete rows present, a subset of complete rows would have been too small for any meaningful clustering analysis to be applied. Not only was it possible to successfully fit a model-based clustering algorithm on this entire dataset, it was also possible to fit different cases of the SMSN family, highlighting the potential of extending the theory of mixtures of scale mixture distributions to handle missing values, thereby providing a variety of component distributions. 

The ECM for FMSMSN with MAR algorithm applies the simplest model, namely the mixture of skew-normal distributions to account for skewness present in incomplete data. In addition, three alternative component distributions are considered, each tailored to reflect differing levels of asymmetry and tail behaviour. These include the skew-slash and skew-t distributions, which naturally simplify to the multivariate skew-normal when appropriate. Notably, the skew-variance-gamma distribution can also reduce to the skew-Laplace. All four models are flexible enough to converge to the standard normal distribution when skewness and excess kurtosis are not present. In essence, the methodology proposed here accommodates asymmetry through distinct distributions, each with its own strengths and degrees of flexibility, allowing for comparative assessment in selecting the most suitable model to describe the resulting clusters.

\section{Conclusion}
\label{sec:Conclusion}


This paper presents a novel framework for clustering incomplete data by extending the finite mixture of multivariate skew-normal scale mixture (FMSMSN) distributions to accommodate missingness under the missing at random (MAR) assumption. The proposed ECM algorithm features closed-form expressions for the E-step and exploits the moments of the multivariate skew-normal distribution, as developed in earlier extensions of the multivariate normal family (\cite{tt1,tt2}).

Simulation studies demonstrate strong performance in both clustering accuracy and parameter recovery. The results vary depending on cluster separation, sample size, and the specific characteristics of each FMSMSN component. This variation highlights a major strength of the extended model family: the availability of multiple distributional forms allows for a flexible balance between parsimony and complexity, depending on the modeling needs.

The application to data on carbon emissions further illustrates these advantages. The proposed approach allows for direct analysis of incomplete datasets without the need for imputation or listwise deletion, and supports comparative evaluation across FMSMSN variants to identify the most appropriate mixture model for a given dataset.

Future work will focus on extending the methodology to accommodate censored data, of which MAR is a special case. This generalisation will broaden the applicability of the framework, including its use in regression modeling under the same flexible family of distributions (\cite{skewtcensor,skewnormalcensor,skewnormalEMcensor}).

\section*{Acknowledgments}

...

Cristina Tortora acknowledges the support by NSF grant number 2209974.
Antonio Punzo acknowledges the support by the Italian Ministry of University and Research (MUR) under the PRIN 2022 grant number 2022XRHT8R (CUP: E53D23005950006), as part of ‘The SMILE Project: Statistical Modelling and Inference to Live the Environment’, funded by the European Union – Next Generation EU.
This work was based upon research supported in part by the National Research Foundation (NRF) of South Africa (SA), ref. SRUG2204203965, UID 119109 (Bekker), the Department of Research and Innovation at the University of Pretoria (SA), as well as the DSI-NRF Centre of Excellence in Mathematical and Statistical Sciences (CoE-MaSS). The opinions expressed and conclusions arrived at are those of the authors and are not necessarily to be attributed to the NRF.

\printbibliography
\appendix
\section{Supplementary material}


\begin{lemma}
\label{lemma1}
    Suppose the vectors and matrix of the quadratic form $\by^{\top}\bA^{-1}\by$ is partitioned as follows:
    \begin{align}
        \by = \begin{pmatrix} \by_1 \\ \by_2\end{pmatrix} \hspace{0.15cm} \text{ and } \hspace{0.15cm}
        \bA = \begin{pmatrix}
                \bA_{11} & \bA_{12}\\
                \bA_{21} & \bA_{22}
              \end{pmatrix}. \nonumber
    \end{align}
    Then 
    \begin{align}\by^{\top}\bA^{-1}\by = \left(\by_1 - \bA_{12}\bA^{-1}_{22}\by_2\right)^{\top}\left(\bA_{11} - \bA_{12}\bA^{-1}_{22}\bA_{21} \right)^{-1} 
    \left(\by_1 - \bA_{12}\bA^{-1}_{22}\by_2\right) + \by^{\top}_2\bA_{22}^{-1}\by_2.
    \end{align}
\end{lemma}

\begin{lemma}
\label{lemma2}
 $\delta_{0c,g}  =\dot{\blam}_{o,g}^{\top}\bSig_{oo,g}^{-1/2}(\bx_{i}^{(o)} - \bmu_{o,g}).$
 \begin{proof}
     From Lemma \ref{lemma1} $\bDelta_{g}^{\top}\bSigg^{-1}\bDelta_{g}$ can be expressed as follows:
     \begin{align}
     \label{quad}
         \bDelta_{g}^{\top}\bSigg^{-1}\bDelta_{g} = \left( \bDelta_{m,g} - \bSig_{mo,g}\bSig_{oo,g}^{-1}\bDelta_{o,g} \right)^{\top}\bSig_{c,g}^{-1}\left( \bDelta_{m,g} - \bSig_{mo,g}\bSig_{oo,g}^{-1}\bDelta_{o,g} \right) +  \bDelta_{o,g}^{\top}\bSig_{oo,g}^{-1}\bDelta_{o,g}.
     \end{align}
     Simplifying $1 + \blam_{c,g}^{\top}\blam_{c,g}$:\\
     \begin{align}
         1 + \blam_{c,g}^{\top}\blam_{c,g} 
         &= \frac{1 - \bDelta_{g}^{\top}\bSigg^{-1}\bDelta_{g} + \left(\bDelta_{m,g} - \bSig_{mo,g} \bSig_{oo,g}^{-1} \bDelta_{o,g} \right)^{\top} \bSig_{c,g}^{-1}\left(\bDelta_{m,g} - \bSig_{mo,g} \bSig_{oo,g}^{-1} \bDelta_{o,g} \right)}{ 1 - \bDelta_{g}^{\top}\bSigg^{-1}\bDelta_{g} } & \nonumber\\
         & = \frac{ 1 - \bDelta_{g}^{\top}\bSigg^{-1}\bDelta_{g} +\bDelta_{g}^{\top}\bSigg^{-1}\bDelta_{g}  - 1 - \bDelta_{o,g}^{\top}\bSig_{oo,g}^{-1}\bDelta_{o,g} }{1 - \bDelta_{g}^{\top}\bSigg^{-1}\bDelta_{g}}& \nonumber\\
         & = \frac{1 - \bDelta_{o,g}^{\top}\bSig_{oo,g}^{-1}\bDelta_{o,g}  }{ 1 - \bDelta_{g}^{\top}\bSigg^{-1}\bDelta_{g}}.
     \end{align}
     Therefore it follows that
          \begin{align}
         \delta_{0c,g} 
         &= \frac{ \lambda_{0,c,g} }{ \sqrt{1 +  \blam_{c,g}^{\top} \blam_{c,g} }  } &\nonumber\\
         &= \frac{ \bDelta_{o,g}^{\top}\bSig_{oo,g}^{-1}(\bx_{i}^{(o)} - \bmu_{o,g}) }{\sqrt{1 - \bDelta_{g}^{\top}\bSigg^{-1}\bDelta_{g}} } \sqrt{\frac{  1 - \bDelta_{g}^{\top}\bSigg^{-1}\bDelta_{g} }{ 1 - \bDelta_{o,g}^{\top}\bSig_{oo,g}^{-1}\bDelta_{o,g}  } }. & \nonumber\\
         &= \frac{\bDelta_{o,g}^{\top}\bSig_{oo,g}^{-1}(\bx_{i}^{(o)} - \bmu_{o,g}) }{\sqrt{ 1 - \bDelta_{o,g}^{\top}\bSig_{oo,g}^{-1}\bDelta_{o,g} } } &\nonumber\\
         & = \dot{\blam}_{o,g}^{\top}\bSig_{oo,g}^{-1/2}(\bx_{i}^{(o)} - \bmu_{o,g}).& \nonumber
     \end{align}
 \end{proof}
\end{lemma}

\begin{lemma}
    \label{lemma3}
     $\frac{\mu_{T_{i,g}}}{\sigma_{T_{g}}} = A_{i,g}^{(o)} = \dot{\blam}_{o,g}^{\top}\bSig_{oo,g}^{-1/2}(\bx_i^o-\bmu_{o,g})$
\end{lemma}
\begin{proof}
From Lemma \ref{lemma2}:
    \begin{align}
    \label{main ratio}
        \frac{\mu_{T_{i,g}}}{\sigma_{T_{g}}} = \left( 1 + \bDelta_{o,g}^{\top} \bOmega_{oo,g}^{-1} \bDelta_{o,g} \right)^{-1/2}\left( \bDelta_{o,g}^{\top} \bOmega_{oo,g}^{-1}(\bx_{i}^o - \bmu_{o,g})\right).
    \end{align}
    From block matrix inversion and the Shermann-Morrison formula \cite{sherman-morrison}:
    \begin{align}
    \label{sherman morrison}
       \bOmega_{oo,g}^{-1} = (\bSig_{oo,g}- \bDelta_{o,g}\bDelta_{o,g}^{\top} )^{-1} = \bSig_{oo,g} + \frac{\bSig_{oo,g}^{-1}\bDelta_{o,g}\bDelta_{o,g}^{\top}\bSig_{oo,g}^{-1} }{1 - \bDelta_{o,g}\bSig_{oo,g}^{-1}\bDelta_{o,g}^{\top}}.
    \end{align}
    Now, $\bDelta_{o,g}^{\top} \bOmega_{oo,g}^{-1}(\bx_{i}^o - \bmu_{o,g})$ simplifies as follows using \eqref{sherman morrison}:
    \begin{align}
    \label{part1}
        \bDelta_{o,g}^{\top} \bOmega_{oo,g}^{-1}(\bx_{i}^o - \bmu_{o,g}) &= \frac{\bDelta_{o,g}^{\top}\bSig_{oo,g}^{-1}(\bx_{i}^o - \bmu_{o,g})(1 - \bDelta_{o,g}^{\top}\bSig_{oo,g}^{-1}\bDelta_{o,g}) }{1 - \bDelta_{o,g}^{\top}\bSig_{oo,g}^{-1}\bDelta_{o,g}} + \frac{\bDelta_{o,g}^{\top}\bSig_{oo,g}^{-1}\bDelta_{o,g} \bDelta_{o,g}^{\top}\bSig_{oo,g}^{-1}(\bx_{i}^o - \bmu_{o,g})}{1 - \bDelta_{o,g}^{\top}\bSig_{oo,g}^{-1}\bDelta_{o,g}} &\\
        &= \frac{\bDelta_{o,g}^{\top}\bSig_{oo,g}^{-1}(\bx_{i}^o - \bmu_{o,g})}{1 - \bDelta_{o,g}^{\top}\bSig_{oo,g}^{-1}\bDelta_{o,g}}.&\nonumber
    \end{align}
    Similarly, $\left( 1 + \bDelta_{o,g}^{\top} \bOmega_{oo,g}^{-1} \bDelta_{o,g} \right)^{-1/2}$ simplifies using \eqref{sherman morrison}:
    \begin{align}
    \label{part2}
        \left( 1 + \bDelta_{o,g}^{\top} \bOmega_{oo,g}^{-1} \bDelta_{o,g} \right)^{-1/2} = \left( 1 - \bDelta_{o,g}^{\top} \bSig_{oo,g}^{-1} \bDelta_{o,g} \right)^{1/2}.
    \end{align}
    Substituting \eqref{part1} and \eqref{part2} into the right-hand side of \eqref{main ratio} simplifies to:
    \begin{align}
        \frac{\left( 1 - \bDelta_{o,g}^{\top} \bSig_{oo,g}^{-1} \bDelta_{o,g} \right)^{1/2} \bDelta_{o,g}^{\top}\bSig_{oo,g}^{-1}(\bx_{i}^o - \bmu_{o,g})  }  { 1 - \bDelta_{o,g}^{\top} \bSig_{oo,g}^{-1} \bDelta_{o,g} } = \dot{\blam}_{o,g}^{\top}\bSig_{oo,g}^{-1/2}(\bx_i^o-\bmu_{o,g}),
    \end{align}
    where $\dot{\blam}_{o,g}  =  \frac{\bSig_{oo,g}^{-1/2}\bDelta_{o,g}}{\left( 1 - \bDelta_{o,g}^{\top} \bSig_{oo,g}^{-1} \bDelta_{o,g} \right)^{1/2}}$.
\end{proof}

\begin{theorem}
\label{t cond dist}
    $T_i|\bX^{o}_i = \bx^{o}_i, U_i = u_i ,Z_{i,g} = 1 \sim TN(\mu_{T_{i,g}}, \ku_i\sigma^2_{T_{g}})$,  where $\sigma^2_{T_{g}} = \left(1 + \bDelta_{o,g}^{\top}\bOmega_{oo,g}^{-1}\bDelta_{o,g} \right)^{-1}$ and $ \mu_{T_{i,g}} = \sigma^2_{T_{g}} \bDelta_{o,g}^{\top}\bOmega_{oo,g}^{-1}(\bx^{(o)}_i - \bmu_{o,g})$.
\end{theorem}
\begin{proof}
    Using Bayes' theorem $f(t_i| \bx^{o}_i, u_i ,z_{i,g}) \propto f(\bx^o_i|t_i, u_i ,z_{i,g})f(t_i|u_i ,z_{i,g})$.
    \newline
    The right hand side of the proportionality is now expanded and simplified. Noting from Theorem \ref{sn cond dist} and from (\ref{sr scale mix}), $\bX^{o}_i|T_i=t_i, U_i = u_i ,Z_{i,g} = 1 \sim SN(\bmu_{o,g}+t_i\bDelta_{o,g}, \bOmega_{oo,g}, \bm{0} )$ and $T_i|U_i=u_i,Z_{i,g}=1 \sim TN(0,\ku_i)$.
    \begin{align}
        f(\bx^o_i|t_i, u_i ,z_{i,g})f(t_i|u_i ,z_{i,g}) &= f_{\text{SN}}(\bx^{o}_i; \bmu_{o,g} +t_i \bDelta_{o,g},\kappa_i\bOmega_{oo,g})f_{\text{TN}}(t_i;0,\ku_i)&.\nonumber\\
       &\propto \mathrm{exp}\left\{-\frac{1}{2}(\bx^{o}_i- \bm{\Delta}_{o,g}t_i - \bm{\mu}_{o,g} )^T 
		\left(\kappa_i\bm{\Omega}_{oo,g} \right)^{-1} 
		(\bx^{o}_i-  \bm{\Delta}_{o,g}t_i - \bm{\mu}_{o,g} ) \right\}\mathrm{exp} 
		\left\{ -\frac{1}{2\ku_i}t_i^2\right\} \nonumber &\\
		&\propto \mathrm{exp} \left\{ -\frac{1}{2\ku_i} \left( (\bm{\Delta}_{o,g})^T\left(\bm{\Omega}_{oo,g} \right)^{-1}\bm{\Delta}_{o,g} t_i^{2} + t_i^2 - 2(\bm{\Delta}_{o,g})^T\left(\bm{\Omega}_{oo,g} \right)^{-1}(\bx_{i}^{o} - \bm{\mu}_{o,g}\right)t_i  \right\}\nonumber&\\
		&\propto \mathrm{exp} \left\{ - \frac{1}{2\ku_i} \left(1 + ((\bm{\Delta}_{o,g})^T\left(\bm{\Omega}_{oo,g}\right)^{-1} \bm{\Delta}_{o,g} \right) 
		\left( t_i^2 - \frac{ (\bm{\Delta}_{o,g})^T\left(\bm{\Omega}_{oo,g} \right)^{-1}(\bx_i^o - \bmu_{o,g})	}{	1 + (\bm{\Delta}_{o,g})^T(\bm{\Omega}_{oo,g}^{-1}) \bm{\Delta}_{o,g}} t_i \right)		\right\} \nonumber &\\
		\label{posterior obs}
		&\propto \mathrm{exp} \left\{ -\frac{1}{2\ku_i}\left(1 + (\bm{\Delta}_{o,g})^T \left(\bm{\Omega}_{oo,g}\right)^{-1} \bm{\Delta}_{o,g} \right)  \left(	t_i - \frac{ (\bm{\Delta}_{o,g})^T\left(\bm{\Omega}_{oo,g}\right)^{-1}(\bx_{i}^o - \bm{\mu}_{o,g})	}{	1 + (\bm{\Delta_{o,g})}^T\left(\bm{\Omega}_{oo,g}\right)^{-1} \bm{\Delta}_{o,g}}	\right)	^2	\right\},
    \end{align}
    which is proportional to a truncated normal distribution with parameters $\ku_i\sigma^2_{T_{g}} = \ku_i\left(1 + \bDelta_{o,g}^{\top}\bOmega_{oo,g}^{-1}\bDelta_{o,g} \right)^{-1}$ and $ \mu_{T_{i,g}} = \sigma^2_{T_{g}} \bDelta_{o,g}^{\top}\bOmega_{oo,g}^{-1}(\bx^{(o)}_i - \bmu_{o,g})$.
\end{proof}
\begin{definition}
    \label{t dist def}
    A random vector $\bX \in \mathbb{R}^p$ follows a t distribution with parameters $\bmu, \bSig,$ and $\nu$ if the pdf is:
    \begin{align}
        f_t(x) = \frac{\Gamma\left( \frac{\nu +p }{2} \right) }{(\pi\nu)^{p/2}\Gamma\left(\frac{\nu}{2} \right)|\bSig|^{p/2}} \left( 1 + \frac{d^2}{\nu} \right)^{-(\nu+p)/2},
    \end{align}
    where $d^2 = (\bx - \bmu)\bSig^{-1}(\bx - \bmu)$; This is denoted $\bX \sim~ t_p(\bmu,\bSig, \nu)$.
\end{definition}

\begin{definition}
    \label{gamma dist def}
    A random variable $X$ follows a gamma distribution with parameters $\alpha$ and $\beta$ if it has the pdf:
    \begin{align}
        f_{\text{GAM}}(x) = \begin{cases}
            \frac{\beta^{\alpha}}{\Gamma(\alpha)}x^{\alpha-1}e^{\beta x} & \text{ if } x>0 \\
            0 & \text{otherwise};
        \end{cases}
    \end{align}
    This is denoted as $X\sim Gam(\alpha, \beta)$.
\end{definition}

\begin{definition}
\label{beta pdf}
    A random variable $X$ follows a beta distribution with parameters $\alpha_1$ and $\alpha_2$ if the pdf is:
    \begin{align}
        f_{\beta}(x) = \begin{cases}
                        \frac{\Gamma(\alpha_1 + \alpha_2)}{\Gamma(\alpha_1)\Gamma(\alpha_2)}x^{\alpha_1 -1}(1-x)^{\alpha_2 -1} & \text{ if  } 0 ~< x ~< 1\\
                        0 & \text{otherwise};\end{cases}
    \end{align}
    This is denoted as $X \sim Beta(\alpha_1,\alpha_2)$.
\end{definition}

\begin{definition}
\label{gh pdf}
    A random variable $X$ follows a generalised hyperbolic distribution with parameters $\mu,\sigma, \eta, \psi$ and $\gamma$ if the pdf is:
    \begin{align}
        f_{\text{GH}}(x) = \begin{cases}
                        \sqrt{\frac{\gamma}{2\pi}} \left(\psi^{\eta} K_{\eta}(\psi\gamma \right)^{-1}(\sqrt{\psi^2 + d^2})^{\eta-1/2}K_{\eta-1/2}\left(\left(\sqrt{\psi^2 + d^2}\right)\gamma \right)  & \text{ if  } x>0 \\
                        0 & \text{otherwise,}\end{cases}
    \end{align}
    where $d^2 = \left(\frac{x-\mu}{\sigma}\right)^2$; this is denoted as $X\sim GH(\mu,\sigma,\eta,\psi,\gamma)$.
\end{definition}

\end{document}